\documentclass[conference, twocolumn, 11pt]{IEEEtran}
\usepackage{amsmath,amssymb,amsfonts, amsthm}
\usepackage{algorithm}
\usepackage{algpseudocode}
\usepackage{graphicx}
\usepackage{textcomp}
\usepackage{xcolor}
\usepackage{subcaption}
\usepackage{url}
\usepackage[authoryear]{natbib}   
\usepackage{graphicx}
\usepackage{mdframed}
\usepackage{graphicx}
\usepackage{xspace}
\usepackage{multirow}
\usepackage[american]{babel}
\usepackage{mathtools} 
\usepackage{booktabs} 
\usepackage{array}
\usepackage{tikz} 
\definecolor{citationgreen}{RGB}{0,125,90}
\usepackage[colorlinks=true,linkcolor=blue,citecolor=citationgreen,urlcolor=magenta]{hyperref}

\graphicspath{{Figures/}}

\tikzset{bmp/.style={
    scale=6,
    line cap=round,
    line join=round
}}

\algtext*{EndIf}
\algtext*{EndFor}
\algtext*{EndWhile}
\algtext*{EndProcedure}
\algtext*{EndFunction}

\newtheorem{proposition}{Proposition}
\newtheorem{theorem}[proposition]{Theorem}
\newtheorem{assumption}{Assumption}

\newtheoremstyle{definitionbold}
  {}{}                    
  {\itshape}           
  {}                      
  {\bfseries}             
  {.}                     
  {.5em}                  
  {\thmname{#1}\thmnumber{ #2}\thmnote{ \textbf{(#3)}}}

\theoremstyle{definitionbold}
\newtheorem{definition}{Definition}

\theoremstyle{definition}

\newtheorem{remark}{Remark}
\newtheorem{example}{Example}

\newcommand{\gcn}{\textsc{Gcn}\xspace}
\newcommand{\gat}{\textsc{Gat}\xspace}

\newcommand{\cora}{\textsc{Cora}\xspace}
\newcommand{\cseer}{\textsc{CiteSeer}\xspace}

\newcommand{\mA}{\mathcal{A}}
\newcommand{\mB}{\mathcal{B}}
\newcommand{\md}{\mathrm{d}}
\newcommand{\mE}{\mathcal{E}}
\newcommand{\mG}{\mathcal{G}}
\newcommand{\mH}{\mathcal{H}}

\newcommand{\mR}{\mathcal{R}}
\newcommand{\mS}{\mathcal{S}}
\newcommand{\mV}{\mathcal{V}}
\newcommand{\bX}{\mathbf{X}}

\newcommand{\bY}{\mathbf{Y}}

\newcommand{\mZ}{\mathcal{Z}}

\tikzset{bmp/.style={
    scale=6,
    line cap=round,
    line join=round
}}

\newcommand{\BMPregion}[4][6]{%
\begin{tikzpicture}[scale=#1, line cap=round, line join=round]

\path[use as bounding box] (-0.10,-0.10) rectangle (1.18,1.18);

\def\eta{#2}
\def\epsl{#3}
\def\epsr{#4}

\pgfmathsetmacro{\a}{(1/\eta)*exp(-\epsl)}
\pgfmathsetmacro{\b}{(\eta)*exp(-\epsr)}
\pgfmathsetmacro{\ab}{\a*\b}
\pgfmathsetmacro{\den}{1-\ab}

\pgfmathsetmacro{\xA}{0}
\pgfmathsetmacro{\yA}{1}
\pgfmathsetmacro{\xB}{(1-\b)/\den}
\pgfmathsetmacro{\yB}{1-\a*\xB}
\pgfmathsetmacro{\xC}{1}
\pgfmathsetmacro{\yC}{0}

\pgfmathsetmacro{\yD}{(exp(-\epsr) - 1/\eta)/(exp(-\epsr) - exp(\epsl))}
\pgfmathsetmacro{\xD}{1 - \eta*exp(\epsl)*\yD}

\draw[thin] (0,0) rectangle (1,1);

\fill[blue!12]
  (\xA,\yA) --
  (\xB,\yB) --
  (\xC,\yC) --
  (\xD,\yD) -- cycle;

\draw[very thick]
  (\xA,\yA) --
  (\xB,\yB) --
  (\xC,\yC) --
  (\xD,\yD) -- cycle;

\draw[->] (-0.06,0) -- (1.12,0) node[below] {$\alpha$};
\draw[->] (0,-0.06) -- (0,1.12) node[left] {$\beta$};
\end{tikzpicture}%
}

\newcommand{\BMPregionInfL}[3][6]{%
\begin{tikzpicture}[scale=#1, line cap=round, line join=round]

\path[use as bounding box] (-0.10,-0.10) rectangle (1.18,1.18);

\def\eta{#2}
\def\epsr{#3}

\pgfmathsetmacro{\b}{\eta*exp(-\epsr)}

\coordinate (A) at (0,1);

\coordinate (B) at ({1-\b},1);

\coordinate (C) at (1,0);

\coordinate (D) at (\b,0);

\draw[thin] (0,0) rectangle (1,1);

\fill[blue!12] (A)--(B)--(C)--(D)--cycle;
\draw[very thick] (A)--(B)--(C)--(D)--cycle;

\draw[->] (-0.06,0) -- (1.12,0) node[below] {$\alpha$};
\draw[->] (0,-0.06) -- (0,1.12) node[left] {$\beta$};

\draw (1,0) ++(0,0.015) -- ++(0,-0.03);

\draw (0,1) ++(0.015,0) -- ++(-0.03,0);

\end{tikzpicture}%
}

\newcommand{\BMPregionInfR}[3][6]{%
\begin{tikzpicture}[scale=#1, line cap=round, line join=round]

\path[use as bounding box] (-0.10,-0.10) rectangle (1.18,1.18);

\def\eta{#2}
\def\epsl{#3}

\pgfmathsetmacro{\a}{(1/\eta)*exp(-\epsl)}

\coordinate (A) at (1,0);

\coordinate (B) at (1,{1-\a});

\coordinate (C) at (0,1);

\coordinate (D) at (0,\a);

\draw[thin] (0,0) rectangle (1,1);

\fill[blue!12] (A)--(B)--(C)--(D)--cycle;
\draw[very thick] (A)--(B)--(C)--(D)--cycle;

\draw[->] (-0.06,0) -- (1.12,0) node[below] {$\alpha$};
\draw[->] (0,-0.06) -- (0,1.12) node[left] {$\beta$};

\draw (1,0) ++(0,0.015) -- ++(0,-0.03);

\draw (0,1) ++(0.015,0) -- ++(-0.03,0);

\end{tikzpicture}%
}

\begin{document}
\title{Bayesian Membership Privacy for Graph Neural Networks}

\thanks{Identify applicable funding agency here. If none, delete this.}

\author{\IEEEauthorblockN{Sinan Yıldırım}
\IEEEauthorblockA{\textit{Faculty of Engineering and Natural Sciences} \\
\textit{Sabancı University}\\
İstanbul, Turkey \\
sinanyildirim@sabanciuniv.edu}
\and
\IEEEauthorblockN{Megha Khosla}
\IEEEauthorblockA{\textit{Intelligent Systems Department} \\
\textit{Delft University of Technology (TU Delft)}\\
Delft, The Netherlands \\
M.Khosla@tudelft.nl}
}

\maketitle

\begin{abstract}
Existing privacy analyses for Graph Neural Networks (GNNs) largely inherit assumptions from non-graph settings, overlooking structural correlations and stochastic training-graph sampling. In particular, node-dependent priors make type-I and type-II errors alone insufficient to characterize the best membership inference test. To address this, we introduce Bayesian Membership Privacy (BMP), a sampling-aware formulation of node-level membership privacy that incorporates node-dependent priors and treats graph sampling probabilities as part of the adversary’s knowledge. BMP casts membership inference as a Bayesian hypothesis test and accordingly quantifies membership privacy in terms of posterior membership probability. We explore theoretical properties of BMP in relation to the existing definitions in the literature. We further propose a practical, sampling-aware auditing mechanism to estimate the parameters of BMP as a measure of node-level privacy leakage in GNNs. We conduct experiments on benchmark graph datasets and show that BMP yields fine-grained privacy insights that are not visible through global attack accuracy alone. 
\end{abstract}

\section{Introduction}
Graph Neural Networks (GNNs) are now a standard modeling paradigm for learning from structured and relational data, with successful applications in domains such as molecular modeling, recommender systems, and knowledge graph reasoning. Alongside their increasing adoption, recent studies have shown that trained GNNs can inadvertently expose information about the data used during training, raising concerns about privacy in graph-based learning~\citep{duddu2020quantifying,conti2022label,he2021stealing,jnaini2022powerful,wang2024subgraph,wu2021adapting,dai2022comprehensive,wang2021membership,olatunji2021membership}. Much of this work draws inspiration from privacy attacks developed for independent data records, even though graph-structured data fundamentally violates independence assumptions.

We study node-level membership privacy in GNNs trained for node classification, where the adversary seeks to determine whether a specific node was included in the training set. Unlike i.i.d.\ settings, nodes in a graph are interdependent through edges and overlapping neighborhoods. As a result, the inclusion of a single node can affect the learned representations and predictions of many others. Empirical implications of this phenomenon on MI were first observed in~\citep{olatunji2021membership}, where membership inference risk did not consistently correlate with the generalization gap. 

Not only do structural dependencies in graphs affect empirical membership inference behavior, but they also make standard differential privacy (DP) notions difficult to apply. In particular, node-level DP \citep{sajadmanesh2022gap, olatunji2023releasing, hegoing} typically assumes a worst-case adversary who knows the entire dataset except for one individual. In graph settings, however, nodes are coupled through edges and shared neighborhoods, and training nodes are selected via a sampling process that interacts with this structure. As a result, some nodes are inherently more likely to be included than others, and knowledge of the remaining graph can already reveal substantial information about a node’s membership. 

Moreover, unlike the i.i.d.\ setting—where the success of membership inference attacks (MIA) can be formally linked to DP guarantees under exchangeability assumptions—no such relationship is known for GNNs, precisely because structural interdependencies break exchangeability between training and non-training nodes \citep{khosla2026does}. We fill this major gap by developing a \emph{Bayesian notion of membership privacy} (BMP) that explicitly incorporates prior inclusion probabilities as a proxy for the data-generating process and quantifies privacy in terms of posterior belief updates. BMP is developed in Section \ref{sec: Bayesian Membership Privacy for Graphs} on the basis that membership inference cannot be treated as a uniform binary hypothesis test applied identically to all nodes. Instead, node-level privacy risk is inherently heterogeneous: nodes differ both in their sampling-induced prior inclusion probabilities and in the strength of the signal available to an adversary after training. Consequently, evaluating privacy leakage solely through global attack accuracy or fixed type-I/type-II error bounds fails to capture the true, node-specific risk. We \emph{model membership inference as a Bayesian hypothesis testing problem with node-dependent priors}.

Bayesian versions of privacy have been defined in several works \citep{Lee_and_Clifton_2012, li2013membership, Yang_et_al_2015, Triastcyn_and_Faltings_2020}. BMP differs from the existing definitions in several ways, the most significant of which are that it is defined directly in terms of posterior probabilities and is asymmetrical. 

In Section \ref{sec: Definitions for membership privacy}, we provide several properties of BMP, such as composition, post-processing, and its relation to pufferfish type privacy. Furthermore, in Section \ref{sec: Statistical Guarantees of BMP} we establish theoretical relations between BMP and the capability of any given MIA that operates on the output of a BMP algorithm. Those relations concern the expected cost and the probability of wrong decision by the adversary, as well as a one-to-one relation between BMP and type-I and type-II error probabilities of the MIA. The last one enables a new method to audit an algorithm by estimating the BMP parameters using MIA results, described in Section \ref{sec: Estimating privacy parameters from attack results}.

Furthermore, in Section \ref{sec: Designing MIAs and measuring their performance}, we design MIAs to estimate BMP privacy parameters for GNN models across different graph sampling strategies, training set sizes, edge settings, and noise levels. The results show that our framework supports fine-grained, uncertainty-aware analysis, revealing how sampling and model design choices shape node-level privacy leakage.

\section{Background}
\subsection{Graph Neural Networks} \label{sec: Graph Neural Networks}
\label{sec:GNNs}

Let $\mathcal{G} = (\mathcal{V}, \mathcal{E})$ denote a graph, where $\mathcal{V}$ is the set of $N$ nodes indexed by integers and $\mathcal{E} \subseteq \mathcal{V} \times \mathcal{V}$ is the set of edges connecting pairs of nodes. Each node $i\in \{1,2,\ldots, N\}$ is associated with an input feature vector $\mathbf{x}_i \in \mathbb{R}^{d}$. Collectively, the node features are represented by the feature matrix $\mathbf{X} \in \mathbb{R}^{N \times d}$, where the $i$-th row corresponds to the feature vector $\mathbf{x}_i$ of node $i$. Let $\mathbf{Y} \in \{0,1\}^{N \times c}$ denote the one-hot encoded label matrix for the training nodes, where $c$ is the number of classes.

GNNs are a class of machine learning models designed to operate on graph-structured data.
In the supervised node classification setting, a GNN is trained using the triplet $(\mathcal{G}, \mathbf{X}, \mathbf{Y}_{\text{train}})$, where the graph structure $\mathcal{G}$ and node features $\mathbf{X}$ define the input, and supervision is provided only through the labels of the training nodes encoded in $\mathbf{Y}_{\mathrm{train}}$.

The core operation in GNNs is graph convolution, where each node representation is updated by aggregating its own features together with the features of its neighbors, followed by a non-linear transformation. Let $\mathbf{x}_i^{(\ell)}$ denote the feature representation of node $i$ at layer $\ell$, with $\mathbf{x}_i^{(0)} = \mathbf{x}_i$.
Let $\mathcal{N}(i) = \{ j \mid (i,j) \in \mathcal{E} \}$ denote the set of 1-hop neighbors of node $i$ induced by the edge set $\mathcal{E}$. The graph convolution operation at layer $\ell$ is defined as:
\begin{align*}
\mathbf{z}_i^{(\ell)} &=
\operatorname{AGG}^{(\ell)}\!\left(
\mathbf{x}_i^{(\ell-1)}, 
\{ \mathbf{x}_j^{(\ell-1)} \mid j \in \mathcal{N}(i) \}
\right),  \nonumber \\
\mathbf{x}_i^{(\ell)} &=
\operatorname{TRANS}^{(\ell)}\!\left(\mathbf{z}_i^{(\ell)}\right).
\end{align*}
Here, the aggregation $\operatorname{AGG}^{(\ell)}$ operator is GNN-specific and defines how neighborhood information is combined, while the transformation operation $\operatorname{TRANS}^{(\ell)}$ typically consists of a learnable linear mapping followed by a non-linear activation function. At the final layer (denoted $L$), a softmax function is applied to obtain the predicted class probabilities for each node:
$\hat{\mathbf{y}}_i = \operatorname{softmax}\!\left(\mathbf{z}_i^{(L)} \mathbf{\Theta}\right),
$
where $\hat{\mathbf{y}}_i \in \mathbb{R}^{c}$ is the predicted class probability vector for node $i$, $\mathbf{\Theta}$ is a learnable weight matrix, and the $j$-th entry $\hat{\mathbf{y}}_{i,j}$ denotes the predicted probability that node $i$ belongs to class $j$. The model is trained by minimizing the cross-entropy loss between $\hat{\mathbf{y}}_i$ and the ground-truth label vector $\mathbf{Y}_{\mathrm{train}}[i]$.


Because GNNs aggregate information from local neighborhoods, a query for node $i$ at inference time requires access not only to its feature vector $\mathbf{x}_i$ but also to the features of its $L$-hop neighbors. Accordingly, a node-level input can be written as
$
\mathcal{I}_i := \bigl(\mathbf{x}_i,\ { \mathbf{x}_j \mid j \in \mathcal{N}^{(L)}(i) } \bigr),
$
where $\mathcal{N}^{(L)}(i)$ denotes the set of nodes within $L$ hops of node $i$. Consequently, two learning regimes are considered for GNNs, namely \emph{transductive and inductive} settings.

In the \emph{transductive} setting, the full graph $\mathcal{G}$ and feature matrix $\mathbf{X}$—including both training and test nodes—are available during training. Inference reuses this fixed graph, so the neighborhood $\mathcal{N}^{(L)}(i)$ of any node is unchanged regardless of which nodes are labeled. Consequently, altering the supervision set does not modify the underlying graph structure. In contrast, in the \emph{inductive} setting, training is performed on a subgraph induced by the training nodes, and inference applies the learned model to previously unseen nodes or graphs. The distinction between these learning regimes is crucial for membership privacy. In the transductive regime, all nodes appear in the training input (with labels provided only for a subset), making node-level membership ill-defined from a privacy perspective. In the inductive regime, however, training and inference operate on disjoint node sets, allowing a meaningful definition of node-level membership and privacy.

When considering the inductive setting, a key question immediately arises: how was the training graph sampled from the underlying universe graph, which may itself be partially or fully available to the adversary.
In this case, membership inference does not reduce to identifying whether a particular node in the given graph was labelled, but instead to determining whether it was selected into the training subgraph under a particular sampling mechanism. Consequently, node-level membership privacy must be defined \emph{with respect to the training graph sampling process}, rather than solely in terms of model access or label supervision.

\subsection{Membership Inference Attacks and Membership Privacy}
Membership privacy \citep{li2013membership} concerns the risk that an adversary can determine whether a particular individual or data record participated in a dataset. Such participation alone may be sensitive, and may also serve as a proxy for more invasive goals such as partial data reconstruction or attribute inference. In machine learning, membership privacy specifically refers to whether an adversary, given access to a trained model and auxiliary information, can infer if a particular data instance was used during training.

MIAs~\citep{shokri2017membership} provide a central empirical framework for assessing membership privacy leakage in machine learning. These attacks demonstrate that an adversary can determine whether a specific data point was used during training by exploiting model outputs. Subsequent work has connected membership risk to overfitting~\citep{yeom2018privacy}, developed stronger and more efficient black-box attacks~\citep{salem2018ml}, and analyzed how confidence scores, loss values, and calibration influence leakage~\citep{carlini2022membership}.

In the context of GNNs, prior studies have empirically shown measurable leakage regarding the inclusion of nodes, edges, or subgraphs in the training data~\citet{olatunji2021membership, he2021stealing, wang2021membership}, confirming that relational learning models are likewise vulnerable to membership inference.

In settings where input instances are sampled independently of one another, the success rate of membership inference attacks provides lower bounds on the privacy budget of DP mechanisms~\citep{yeom2018privacy,carlini2022membership,jayaraman2019evaluating,humphries2023investigating}. Such a connection relies critically on the assumption of statistical exchangeability between member and non-member data points \citep{humphries2023investigating}. Exchangeability means that the joint distribution of the dataset remains unchanged when a single data point is replaced by another point drawn from the same distribution. As a consequence, the joint probability of observing a training set together with a candidate member and a candidate non-member is invariant to their ordering. In other words, from the perspective of the data-generating distribution, a member and a non-member are statistically indistinguishable. Such a symmetry implies that the data distribution itself provides no information that could help an adversary decide which of the two points was used during training. Any successful membership inference must therefore rely solely on differences induced by the learning algorithm.

However, no formal link has yet been established between MIA success and DP guarantees for GNNs. For node-level attacks, a key challenge is that nodes are connected and share neighborhoods, so training and test nodes are not always guaranteed to be statistically exchangeable, especially in the inductive settings \citep{khosla2026does}. 

\section{Bayesian Membership Privacy} \label{sec: Bayesian Membership Privacy for Graphs}

\subsection{Setting and the probabilistic model} \label{sec: Setting and the probabilistic model}
Given a graph $\mG = (\mV, \mE)$ with node feature and label matrices $\mathbf{X}$ and $\mathbf{Y}$ as defined in Section~\ref{sec:GNNs}, we consider a randomized learning algorithm $\mA = (\mS, \Phi)$ operating in two stages. First, a subset $S \subseteq \mV$ is sampled according to a stochastic subset selection mechanism $\mS$, inducing a training subgraph $\mG_S$ and $\mathbf{X}_{S}, \mathbf{Y}_{S}$ accordingly. Second, a training procedure $\Phi$ is applied to $(\mG_S, \mathbf{X}_S, \mathbf{Y}_S)$ to produce an output $W \in \mathcal{W}$. The training procedure $\Phi$ may itself be randomized. We construct the corresponding probabilistic model as follows.

\paragraph{Prior membership probabilities} For each node $i \in \{1, \ldots, N\}$, we define the membership variable
$M_i := \mathbb{I}(i \in S),$
where $S \subseteq \mV$ is the sampled training set. Thus, $M_i \in \{0,1\}$ indicates whether node $i$ was selected into the training subgraph. Let $\gamma_i$ denote the marginal distribution of $M_i$, defined as
\[
\gamma_i(m) := \Pr(M_i = m), \quad m \in \{0,1\}.
\]
We also define conditional membership probabilities. For any vector
$m_{1:N} \in \{0,1\}^N$, let
\[
\gamma_i^{c}(m_i \mid m_{-i}):= \Pr(M_i = m_i \mid M_{-i} = m_{-i}).
\]
where $M_{-i} := (M_{1:i-1}, M_{i+1:N})$ (similarly for $m_{-i}$).

\paragraph{Likelihood of output} We assume that, for any $m_{1:N} \in \{0, 1\}^{N}$, the conditional distribution of $W$ given $M_{1:N} =m_{1:N}$ has a probability density $f(\cdot | m_{1:N})$ with a common measure $\md w$ such that, for any $O \subseteq \mathcal{W}$, 
\[
\Pr(W \in O | M_{1:N} = m_{1:N}) = \int_{O} f(w | m_{1:N}) \md w.
\]
This allows us to define
\[
f_{i}(w | m_{i}) := \sum_{m_{-i}} f(w | m_{1:N}) \Pr(M_{-i} = m_{-i} | M_{i} = m_{i}) 
\]
to be the density of the conditional probability distribution of $W$ given $M_{i} = m_{i}$. Moreover, we will denote the marginal density of $W$ as $f_{0}$, and note that for any $i \in \{1, \ldots, N \}$,
\[
f_{0}(w) = f_{i}(w | 1) \Pr(M_{i} = 1) + f_{i}(w | 0) \Pr(M_{i} = 0).
\]

\paragraph{Posterior membership probabilities} We define posterior probabilities conditional on $\mA$'s output. For any $i \in \{1, \ldots, N\}$ and $m_{1:N} \in \{0, 1\}^{N}$, and $w \in \mathcal{W}$ such that $f_{0}(w) > 0$, define marginal and full conditional posteriors
\begin{align*}
&\pi_{i}(m_{i} | w) = \Pr(M_{i} = m_{i} | W = w),\\
&\pi_{i}^{c}(m_{i} | m_{-i}, w) := \Pr(M_{i} = m_{i} | M_{-i} = m_{-i}, W = w).
\end{align*}

\subsection{Defining membership privacy} \label{sec: Definitions for membership privacy}
The definitions in this section are stated under the setting introduced in Section~\ref{sec: Setting and the probabilistic model}. Regarding the adversary's knowledge \emph{before} observing the output of the algorithm $\mA$, we assume the worst-case scenario.

\begin{assumption}[Adversary's knowledge] \label{assmp: Adversarys knowledge}
The adversary knows the graph $\mG = (\mV, \mE)$, the node feature matrix $\mathbf{X}$, the ground truth label matrix $\mathbf{Y}$ and the algorithm $\mA = (\mS, \Phi)$. The adversary observes the algorithm's output $W$. However, the adversary does \emph{not} know the sampled subset $S$, equivalently the membership vector $M_{1:N}$.
\end{assumption}

Membership alone can indeed reveal sensitive information, even if $\mathbf{X}$ is known. We give two examples where membership and non-membership, respectively, can be inherently sensitive:

\begin{example}[Membership is sensitive] Consider a social graph such as Facebook’s. Suppose a platform trains a model on a subset of users selected for a mental-health intervention classifier. Publicly observable features may include friendship graph, age range, etc. Assume those are known to the adversary, as in Assumption \ref{assmp: Adversarys knowledge} (hence constitute $\mathbf{X}$). However, the training subset consists of users internally flagged as likely vulnerable based on private signals unavailable publicly (e.g., help-center interactions, deleted posts), which are not in $\mathbf{X}$. Then learning that the user $i$ belonged to the training subset does not merely reveal Facebook’s pipeline. It reveals that user $i$ was considered part of a sensitive target population, which is a privacy breach about the individual.
\end{example}

\begin{example}[Non-membership is sensitive] 
Suppose a company trains an internal model on applicants shortlisted for executive leadership potential. Then, inferring that an employee is a non-member may reveal that they were not considered for promotion, which can be highly sensitive professionally.
\end{example}

We begin with a non-Bayesian definition of membership privacy. 
\begin{definition}[Membership privacy (MP)] \label{def: MP}
We say $\mA$ is $\varepsilon$-MP if for any $w \in \mathcal{W}$ such that $f_{0}(w) > 0$, and for any $i \in \{1, \ldots, N\}$, we have $
e^{-\varepsilon} \leq f_{i}(w | 1)/f_{i}(w | 0) \leq e^{\varepsilon}.
$
\end{definition}
Definition \ref{def: MP} focuses on the likelihood ratio given the membership information of a single node, and it can be viewed as a pufferfish-type privacy definition \citep{Kifer_and_Machanavajjhala_2014, Kifer_and_Machanavajjhala_2025} specific to membership information. Therefore, it is less strict than the classical DP definition, which bounds the ratio between likelihoods given the entire memberships differing by one entry. For reference, we provide the DP version below.
\begin{definition}[Membership DP (MDP)] \label{def: MDP}
We say $\mA$ is $\varepsilon$-MDP if for any $w \in \mathcal{W}$ such that $f_{0}(w) > 0$, and for any $i \in \{1, \ldots, N\}$ and $m_{-i} \in \{0, 1\}^{N-1}$, we have 
\[
e^{-\varepsilon} \leq \frac{f(w | m_{1:i-1}, 1, m_{i+1:N})}{f(w | m_{1:i-1}, 0, m_{i+1:N})} \leq e^{\varepsilon}.
\]
\end{definition}
The requirement in Definition~\ref{def: MDP} must hold uniformly for all configurations $m_{-i}$ and is therefore strictly stronger. While such distribution-independent guarantees are attractive in classical i.i.d.\ settings, they can be overly restrictive in graph-structured domains where memberships are statistically dependent. Enforcing worst-case bounds over arbitrary membership configurations may either severely degrade utility or fail to meaningfully reflect how leakage arises through structural correlations.

Even the likelihood-based notion in Definition~\ref{def: MP} abstracts away prior membership probabilities. In relational data, however, prior inclusion probabilities are shaped by the training-graph sampling process and may themselves encode information relevant to privacy. Since neither MP nor MDP accounts for the data-generating distribution, leakage induced by structural dependencies is not fully captured. We, therefore, adopt a Bayesian perspective, where privacy is defined relative to explicit inclusion priors that serve as a proxy for the data-generating process. 

\begin{definition}[Bayesian Membership Privacy (BMP)] \label{def: BMP}
We say $\mA$ is $(\varepsilon_{L}, \varepsilon_{R})$-BMP if any $w \in \mathcal{W}$ such that $f_{0}(w) > 0$, and for any $i \in \{1, \ldots, N\}$, we have 
\begin{equation} \label{eq: BMP inequality}
e^{-\varepsilon_{L}} \leq \frac{\pi_{i}(1 | w)}{\pi_{i}(0 | w)} = \frac{\gamma_{i}(1) f_{i}(w | 1)}{\gamma_{i}(0) f_{i}(w | 0)} \leq e^{\varepsilon_{R}}.
\end{equation}
\end{definition}
The BMP-R and BMP-L bear similarities with the positive posterior membership privacy of \citet{li2013membership}, which we discuss in Appendix \ref{sec: Related definitions}.
Apart from the inclusion of the prior membership probabilities, another feature of Definition \ref{def: BMP} is that the inequality \eqref{eq: BMP inequality} is \emph{asymmetric}. This is deliberately done for flexibility and to be able to address certain issues regarding the meaning of privacy. The most important of those issues is the significance of the knowledge of membership relative to the knowledge of non-membership. In many cases, knowledge of non-membership is not a privacy breach, and one should only focus on the knowledge of membership by the adversary. The following one-sided versions of BMP correspond to the two extreme cases of such an imbalance.
\begin{definition}[Right-sided and Left-sided BMP] \label{def: BMP-R, BMP-L}
$\mA$ is called $\varepsilon$-BMP-R if it is $(\infty, \varepsilon)$-BMP, that is, for all $i = 1, \ldots, N$ and $w \in \mathcal{W}$ such that $f_{0}(w) > 0$, 
\[
\pi_{i}(1 | w)/\pi_{i}(0 | w) \leq e^{\varepsilon}.
\]
Similarly, $\mA$ is called $\varepsilon$-BMP-L if it is $(\varepsilon, \infty)$-BMP.
\end{definition}

Just like the ``differential" version of MP, a ``differential'' version of BMP can be defined, which limits the attacker's ability to discover the membership of a single node given the membership information of every other node. We provide a definition for completeness, but note that this formulation is often impractical in graph-based models, since knowledge of the sampling mechanism and the remaining memberships can already reveal a node’s inclusion independently of the learning algorithm.
\begin{definition}[Bayesian MDP (BMDP)] \label{def: BMDP}
$\mA = (\mS, M)$ is $(\varepsilon_{L}, \varepsilon_{R})$-BMDP if for every $i \in \{1, \ldots, N\}$, $m_{-i} \in \{0, 1\}^{N-1}$, output $w \in \mathcal{W}$ with $f_{0}(w) > 0$, 
\begin{equation} \label{eq: BGDP definin inequality}
e^{-\varepsilon_{L}} \leq \pi_{i}^{c}(1 | m_{-i}, w)/\pi_{i}^{c}(0 | m_{-i}, w)  \leq e^{\varepsilon_{R}}
\end{equation}
\end{definition}
The special cases of $\varepsilon$-BMDP-L and $\varepsilon$-BMDP-R can also be defined similarly as $(\varepsilon, \infty)$-BMDP and $(\infty, \varepsilon)$-BMDP. In fact, $\varepsilon$-BMDP-R corresponds to the differential identifiability of \citet{Lee_and_Clifton_2012}.

Table \ref{tab:privacy-definitions} summarizes the membership privacy definitions given in this section. 
\begin{table*}[t]
\centering
\normalsize
\renewcommand{\arraystretch}{1.35}
\setlength{\tabcolsep}{4pt}
\begin{tabular}{@{}>{\centering\arraybackslash}m{0.05\textwidth} >{\raggedright\arraybackslash}m{0.18\textwidth} >{\centering\arraybackslash}m{0.13\textwidth} >{\raggedright\arraybackslash}m{0.29\textwidth} >{\raggedright\arraybackslash}m{0.28\textwidth}@{}}
\toprule
Defn & \textbf{Name} & \textbf{Parameters} & \textbf{Explanation} & \textbf{Condition} \\
\midrule
\ref{def: MP} & MP & $\varepsilon$ &
LR for node $i$ being in vs. out & $e^{-\varepsilon}
\le
\frac{f_i(w \mid 1)}{ f_i(w \mid 0)}
\le
e^{\varepsilon}
$
\\ \midrule
\ref{def: MDP} & MDP & $\varepsilon$ &
DP-like, LR for node $i$ given $m_{-i}$ &
$e^{-\varepsilon}
\le
\frac{f(w \mid  m_{1:i-1}, 1, m_{i+1:N})}{
f(w \mid m_{1:i-1}, 0, m_{i+1:N})}
\le
e^{\varepsilon}$
\\ \midrule
\ref{def: BMP} & BMP & $(\varepsilon_L,\varepsilon_R)$ &
Posterior odds of membership &
$
e^{-\varepsilon_L}
\le
\frac{\pi_i(1 \mid w)}{\pi_i(0 \mid w)}
\le
e^{\varepsilon_R}
$
\\ 
\ref{def: BMP-R, BMP-L} & BMP-R, BMP-L & $\varepsilon$ &
One-sided BMP &
BMP-R: $(\infty,\varepsilon)$-BMP;

BMP-L: $(\varepsilon,\infty)$-BMP
\\ \midrule
\ref{def: BMDP} &  BMDP & $(\varepsilon_L,\varepsilon_R)$ &
DP-like, posterior odds given $m_{-i}$ &
$
e^{-\varepsilon_L}
\le
\frac{\pi_i^c(1 \mid m_{-i},w) }{ \pi_i^c(0 \mid m_{-i},w)}
\le
e^{\varepsilon_R}
$
\\ 
- & BMDP-R, BMDP-L & $\varepsilon$ &
One-sided BMDP &
BMDP-R: $(\infty,\varepsilon)$-BMDP;

BMDP-L: $(\varepsilon,\infty)$-BMDP
\\
\hline
\end{tabular}
\caption{Summary of membership privacy definitions.}
\label{tab:privacy-definitions}
\renewcommand{\arraystretch}{1}
\setlength{\tabcolsep}{6pt}
\end{table*}

\subsection{Properties of BMP} \label{sec: Properties of Bayesian Membership Privacy}
We focus on BMP in Definition \ref{def: BMP} and study its properties. The proofs of the propositions are in Appendix \ref{sec: Missing Proofs}. First, we show that the BMP parameters $\varepsilon_{L}, \varepsilon_{R}$ are lower-bounded by the ratio of priors. Let $\eta_{i} := \gamma_{i}(1)/\gamma_{i}(0)$ and $\eta_{\max} := \max \{ \eta_{1}, \ldots, \eta_{N} \}$ and $\eta_{\min} = \min \{ \eta_{1},\ldots, \eta_{N}\}$.
\begin{proposition}[Bounds by membership prior] \label{prop: lower bound on BMP}
If $\mA $ is $(\varepsilon_{L}, \varepsilon_{R})$-BMP, $e^{-\varepsilon_{L}} \leq \eta_{i} \leq e^{\varepsilon_{R}}$ for all $i \in \{1, \ldots, N\}$. That is, $\varepsilon_{R} \geq  \log \eta_{\max} $ and $\varepsilon_{L} \geq - \log \eta_{\min}$.
\end{proposition}
This proposition highlights the importance of the sampling mechanism in the BMP definition. Indeed, even a completely private training algorithm, for example, one that discards all input and outputs pure noise, cannot undo leakage that arises solely from the sampling process.

BMP satisfies the post-processing property, which makes it suitable for complex graph learning pipelines where model outputs might be reused (for example, for generating explanations) or transformed (for example, for ranking nodes).
\begin{proposition}[Post-processing] \label{prop: post-processing}
Let $\varphi: \mathcal{W} \to \mZ$ be a measurable function. If $\mA = (\mS, \Phi)$ is $(\varepsilon_{L}, \varepsilon_{R})$-BMP, then $\mathcal{B} =(\mS, \Phi')$ where $\Phi'= \varphi \circ \Phi$ is $(\varepsilon_{L}, \varepsilon_{R})$-BMP, too.
\end{proposition}

A relation between MP and BMP can also be established.
\begin{proposition}[Relation between MP and BMP] \label{prop: MP to MPP}
If $\mA$ is $\varepsilon$-MP, then it is also $(\varepsilon_{L}, \varepsilon_{R})$-BMP with
\[
\varepsilon_{L} = \varepsilon -  \log \eta_{\min}, \quad \varepsilon_{R} = \varepsilon +  \log \eta_{\max}.
\]
Furthermore, if $\mA$ is $(\varepsilon_{L}, \varepsilon_{R})$-BMP, it is also $\varepsilon$-MP where 
\[
\varepsilon = \max\left\{\varepsilon_{L} + \log \eta_{\max},  \varepsilon_{R} - \log \eta_{\min} \right\}.
\]
\end{proposition}
\paragraph{Composing BMP with MP} 
A composition property exists for BMP when combined with MP, allowing privacy guarantees to accumulate in a controlled manner across multiple data accesses.
\begin{proposition}[Composition: BMP-MP]\label{prop: Composition with classical MP}
Let $\mA_{0} = (\mS_{0}, \Phi_{0})$ be an $(\varepsilon_{L}, \varepsilon_{R})$-BMP algorithm and for $t = 1, \ldots, T$, let each $\mA_{t} =(\mS_{t}, \Phi_{t})$ be an $\varepsilon_{t}$-MP algorithm. Then, an algorithm that samples $S \sim \mS_{0}$ and reveals $W_{t} = \Phi_{t}(\mG_{S}, X_{S}, Y_{S})$, $t = 0, \ldots, T$ is $(\varepsilon_{L}+\sum_{t = 1}^{T}\varepsilon_{t}, \varepsilon_{R} + \sum_{t = 1}^{T} \varepsilon_{t})$-BMP. 

Furthermore, this result holds when $W_{0}, \ldots, W_{T}$ run sequentially and $W_{t}$ depends on $W_{0:t-1}$ for $t = 1, \ldots, T$, provided that $\mathcal{A}_{t}$ is $\varepsilon_{t}$-MP for any $W_{0:t-1} = w_{0:t-1}$.
\end{proposition}

Composing BMP with MP corresponds to scenarios where (i) the training data is sampled once, followed by (ii) multiple private operations on the sampled training data. To give an example, consider running multiple $T+1$ epochs of SGD on the sampled training data. We can view this as a composition of the BMP algorithm (sample $S \sim \mathcal{S}$ and run one epoch of noisy SGD), followed by $T$ (MP) algorithms ($T$ more epochs of noisy SGD).  As another example, consider running several private algorithms (such as different variants of GNN algorithms) on the same sampled data. The sampling step can be incorporated into the analysis of the first released algorithm, yielding a BMP-type bound for the combined sample-and-train mechanism. However, once the sampled dataset is reused for additional algorithmic releases, those releases should be treated as further analyses of the same sample and their contribution should be accounted for through their MP bounds under composition.

\paragraph{Composing BMP with BMP} 
Composing multiple (say, $T$) BMP mechanisms is conceptually different. That case corresponds to a case where one performs (sampling, training) $T$ times. Then, from the adversary’s point of view, the question of interest may be whether an individual’s data was used in \emph{any} of those samples. It is possible to provide a composition result for this case, too. 

\begin{proposition}[Composition: BMP-BMP] \label{prop: Composition of BMPs}
Assume that algorithm $\mathcal{A}_{t} = (\mS_{t}, \Phi_{t} )$ is $(\varepsilon_{L}^{(t)}, \varepsilon_{R}^{(t)})$-BMP, for $t = 1, \ldots, T$, and the sampling operations are independent. Define the composition algorithm $\mathcal{A}: (\mS, \Phi)$, where $\mS =\bigotimes_{t=1}^{T} \mS_{t}$ is the product law and 
\[
\Phi(S_{1:T}) = (\Phi_{1}(\mG_{S_{1}}, \bX_{S_{1}}, \bY_{S_{1}}), \ldots, \Phi_T(\mG_{S_{T}}, \bX_{S_{T}}, \bY_{S_{T}}))
\]
applies $\Phi_{t}$ on the sampled dataset according to $S_{t} \sim \mS_{t}$. Then $\mathcal{A}$ is $(\bar{\varepsilon}_{L}, \bar{\varepsilon}_{R})$ BMP, where
\begin{align*}
\bar{\varepsilon}_{L} &= \sum_{t = 1}^{T} \log \mu_{t} - \log \left[1 - \prod_{t = 1}^{T} \mu_{t} \right] \\
\bar{\varepsilon}_{R} &= \log \left[1 - \prod_{t = 1}^{T} (1-\kappa_{t}) \right] - \sum_{t = 1}^{T} \log (1-\kappa_{t}),
\end{align*}
where $\mu_{t} = e^{\varepsilon_{L}^{(t)}}/(1 + e^{\varepsilon_{L}^{(t)}})$ and $\kappa_{t} = e^{\varepsilon_{R}^{(t)}}/(1 + e^{\varepsilon_{R}^{(t)}})$.
\end{proposition}

\subsubsection{Statistical Guarantees of BMP} \label{sec: Statistical Guarantees of BMP}
In this section, we explain the statistical guarantees of BMP in Definition \ref{def: BMP}, in terms of the performance of any MIA that uses the output of $\mA$ as its input. First, we formalize an MIA related to $\mA$.

\begin{definition}[MIA -- Hypothesis Testing Perspective]
\label{def: MIA}
Let $\mG = (\mV, \mE)$ be a graph with node feature matrix $\mathbf{X}$ and label matrix $\mathbf{Y}$. Let $\mA = (\mS, \Phi)$ be a (possibly randomized) algorithm operating on $(\mG, \mathbf{X}, \mathbf{Y})$, producing an output $\hat{Y} \in \mathcal{W}$. For $i \in \{1, \ldots, N \}$, a MIA, denoted as $\mathcal{M}(i, D, \mA)$, is a hypothesis test for the hypotheses
\[
H_{0}: M_{i} = 0, \quad H_{1}: M_{i} = 1,
\]
whose decision rule $D: \mathcal{W} \mapsto \{0, 1\}$ maps the output of $\mA$ to a binary decision in $\{0, 1\}$. 
\end{definition}
Define the type-I and type-II errors of a MIA $\mathcal{M}(i, D, \mA)$,
\begin{align*}
\alpha_{i} &= \Pr(D(W) = 1 | M_{i} = 0), \\
\beta_{i} &= \Pr(D(W) = 0 | M_{i} = 1).
\end{align*}
We have the following crucial result for $(\alpha_{i}, \beta_{i})$.
\begin{theorem}[BMP and error probabilities of MIA] \label{thm: R for BMP}
$\mA$ is $(\varepsilon_{L}, \varepsilon_{R})$-BMP if and only if, for any decision rule $D$ and any $i \in \{1, \ldots, N\}$, the MIA $\mathcal{M}(i, D, \mA)$ satisfies $(\alpha_{i}, \beta_{i}) \in \mR_{\eta_{i}}(\varepsilon_{L}, \varepsilon_{R})$, where for $\varepsilon_{L}, \varepsilon_{R}, \eta \geq 0$,
\begin{align} \label{eq: R for BMP}
\mR_{\eta}(\varepsilon_{L}, \varepsilon_{R}) = \left\{ (\alpha, \beta):\begin{matrix}  \alpha, \beta \in [0, 1] \\  \alpha + \eta e^{\varepsilon_{L}} \beta \in [1, \eta e^{\varepsilon_{L}}] \\ \alpha + \eta e^{-\varepsilon_{R}}  \beta \in [\eta e^{-\varepsilon_{R}}, 1]
\end{matrix} \right\}
\end{align}
\end{theorem}


\begin{figure}
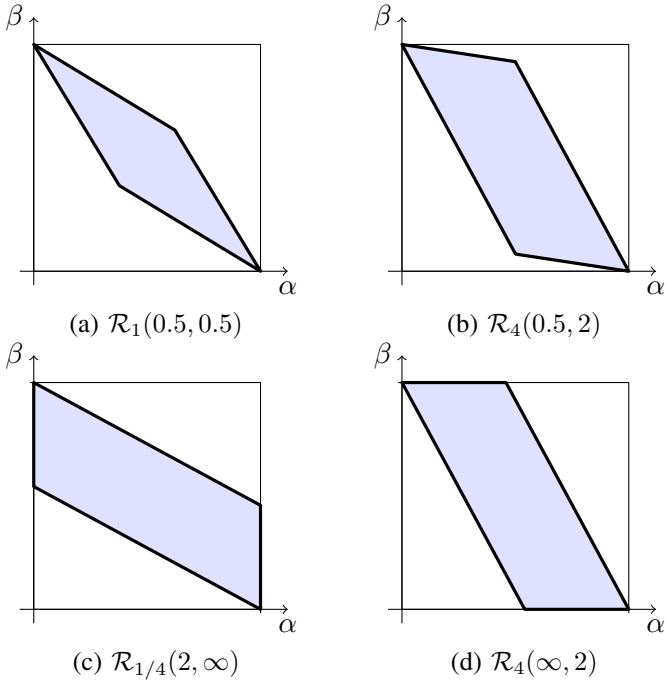

\centering
\begin{subfigure}[t]{0.45\linewidth} 
\centering
\BMPregion[3]{1}{0.5}{0.5}
\caption{$\mR_{1}(0.5, 0.5)$}
\end{subfigure}
\hfill
\begin{subfigure}[t]{0.45\linewidth} 
\centering
\BMPregion[3]{4}{0.5}{2}
\caption{$\mR_{4}(0.5, 2)$}
\end{subfigure}
\hfill
\begin{subfigure}[t]{0.45\linewidth} 
\centering
\BMPregionInfR[3]{0.25}{2}
\caption{$\mR_{1/4}(2, \infty)$}
\end{subfigure}
\hfill
\begin{subfigure}[t]{0.45\linewidth} 
\centering
\BMPregionInfL[3]{4}{2}
\caption{$\mR_{4}(\infty, 2)$}
\end{subfigure}
\caption{The region $\mR_{\eta}(\varepsilon_{L}, \varepsilon_{R})$ for different choices.}
\label{fig: R}
\end{figure}


See Figure \ref{fig: R} for an illustration of $\mR_{\eta}(\varepsilon_{L}, \varepsilon_{R})$ for different choices of $\varepsilon_{L}, \varepsilon_{R}, \eta$ and note the allowed asymmetry in the region. This prior-unaware and symmetric version $\mR_{1}(\varepsilon, \varepsilon)$ is well known \citep[Theorem 2.1]{Kairouz_et_al_2017} and has been used for empirical privacy estimation \citep{Jagielski_et_al_2020, Nasr_et_al_2021, Zanella-Beguelin_et_al_2023, Yildirim_et_al_2025}. The definition in \eqref{eq: R for BMP} generalizes $\mR_{\eta}(\varepsilon_{L}, \varepsilon_{R})$ significantly. Even more importantly, the one-to-one relation in Theorem \ref{thm: R for BMP} will be essential for the method in Section \ref{sec: Estimating privacy parameters from attack results} for empirical estimation of the privacy parameters $\varepsilon_{L}, \varepsilon_{R}$ of BMP. Note that, differently from the prior-unaware classical DP-related definitions, the allowed region of type-I and type-II error probabilities \emph{vary} across $i \in \{1, \ldots, N\}$ according to their prior membership probabilities. 

By the Neyman-Pearson lemma, the likelihood ratio test is known as the best test in terms of type I and type II errors. However, in the Bayesian framework, type I and type II errors are not sufficient alone to characterize the capability of the adversary; one can gain additional insight into the adversary's capability by studying the expected cost of a wrong decision and the probability of a wrong decision as a special case.  For $j = 0, 1$, let $c_{j}$ be the cost of an error under $H_{j}: M_{i} = j$. Then, the overall cost can be defined as 
\begin{align*}
\text{Cost} = c_{0}(1-M_{i}) D(W) +c_{1} M_{i} (1-D(W)).
\end{align*}
The next theorem provides a lower bound on the expected cost of error and, as a special case, the probability of a wrong decision, $D(W) \neq M_{i}$.

\begin{theorem}[BMP and Expected Cost of Error of MIA] \label{thm: expected cost and wrong decision prob}
If $\mA$ is $(\varepsilon_{L}, \varepsilon_{R})$-BMP, then for any decision rule $D$ and any $i \in \{1, \ldots, N\}$, expected cost of the MIA $\mathcal{M}(i, D, \mA)$ is lower-bounded by
\begin{align} \label{eq: lower bound for cost}
 &\mathbb{E}(\textup{Cost})  \geq \max\bigg\{ \min\left( \frac{c_{0}}{1 + e^{\varepsilon_{R}}},   \frac{c_{1}}{1 + e^{\varepsilon_{L}}} \right),  \nonumber\\ 
&\quad\quad \gamma_{i}(0) \min\left(\frac{c_{1}}{ e^{\varepsilon_{L}}},c_{0}\right) , \gamma_{i}(1) \min\left(\frac{c_{0} }{e^{\varepsilon_{R}}}, c_{1} \right) \bigg\}.
\end{align}
In particular, the wrong decision probability is bounded as
\begin{align} \label{eq: lower bound for wrong decision probability}
\Pr(M_{i} \neq D) \geq  \max\left\{ \frac{1}{1 + e^{ \max\{\varepsilon_{L}, \varepsilon_{R}\} }},
\frac{\gamma_{i}(0) }{e^{\varepsilon_{L}}},  \frac{\gamma_{i}(1) }{e^{\varepsilon_{R}}} \right\}.
\end{align}
\end{theorem}

We provide a detailed comparison of BMP with related definitions in Appendix \ref{sec: Related definitions}.

\section{Estimating privacy parameters} \label{sec: Estimating privacy parameters from attack results}

We now describe the Bayesian estimation method that estimates $(\varepsilon_{L}, \varepsilon_{R})$ of BMP for $\mA$ based on the results of a set of MIAs. We adopt a Bayesian approach, as proposed in \citep{Zanella-Beguelin_et_al_2023, Yildirim_et_al_2025} and adapt it to our new privacy definition. In Section \ref{sec: Joint probability model and Bayesian estimation via MCMC}, we present the joint probabilistic model for the random variables involved in a set of MIA experiments and their estimation using an MCMC approach. In Section \ref{sec: Designing MIAs and measuring their performance}, we design a practical MIA for GNNs and provide a computationally efficient method to assess its performance.

\subsection{Joint Probabilistic model} \label{sec: Joint probability model and Bayesian estimation via MCMC}
Suppose we have MIA results for a set of nodes $\mathcal{T} = \{\nu_{1}, \ldots, \nu_{K} \} \subseteq \{1, \ldots, N\}$ are chosen where $K = |\mathcal{T}|$. For each $k$, let $(\alpha_{k}, \beta_{k})$ be the type-I and type-II errors of the designed MIA, $\mathcal{M}(\nu_{k}, D, \mA)$. By Theorem \ref{thm: R for BMP}, 
\[
 \text{$\mA$ is $(\varepsilon_{L}, \varepsilon_{R})$-BMP}  \Leftrightarrow (\alpha_{k}, \beta_{k}) \in \mR_{\eta_{\nu_{k}}}(\varepsilon_{L}, \varepsilon_{R}).
\] 
Using this relation, we can model the probability distribution of $(\alpha_{k}, \beta_{k})$ as
\begin{equation} \label{eq: uniform over R}
(\alpha_{k}, \beta_{k}) | \varepsilon_{L}, \varepsilon_{R} \sim \text{Uniform}\left(\mR_{\eta_{\nu_{k}}}(\varepsilon_{L}, \varepsilon_{R})\right)
\end{equation}
The area of $\mR_{\eta_{\nu_{k}}}(\varepsilon_{L}, \varepsilon_{R})$ can be calculated with geometric manipulations; see Figure \ref{fig: R}. For each node $k$, assume that MIA is applied $n_{0, k}$ times under $H_{0}: M_{\nu_{k}} = 0$ and $n_{1, k}$ times under $H_{1}: M_{\nu_{k}} = 1$, and we obtain the respective counts $A_{k}, B_{k}$ of type I and type II errors. The conditional distribution of $A_{k}, B_{k}$ can generically be shown to have the probability distribution
\begin{equation} \label{eq: conditional distribution of errors}
g_{\theta}(a_{k}, b_{k} | \alpha_{k}, \beta_{k}) = \Pr(A_{k} = a_{k}, B_{k} = b_{k} | \alpha_{k}, \beta_{k})
\end{equation}
where $\theta$ contains whatever (if any) parameters are used to construct $g_{\theta}$, which may be assumed unknown and estimated together with $(\varepsilon_{L}, \varepsilon_{R})$, and we have suppressed $n_{0, k}, n_{1, k}$ for the sake of notational simplicity.

A simple choice is to model $A_{k}, B_{k}$ as independent Binomial random variables. Alternatively, provided that $n_{0, j}$ and $n_{1, j}$ are both large enough, a bivariate normal approximation can be applied to also account for the correlation between $A_{k}, B_{k}$, see \citep{Yildirim_et_al_2025} for details. Combining \eqref{eq: uniform over R} and \eqref{eq: conditional distribution of errors}, and assuming an uninformative prior $p(\varepsilon_{L}, \varepsilon_{R})$ on $(\varepsilon_{L}, \varepsilon_{R})$, and $p(\theta)$ on the parameters $\theta$ of the error count probabilities,  we have the joint probability distribution of $(\varepsilon_{L}, \varepsilon_{R}, \alpha_{1:K}, \beta_{1:K}, A_{1:K}, B_{1:K})$
\begin{align}
& p (\varepsilon_{L}, \varepsilon_{R}, \theta, \alpha_{1:K}, \beta_{1:K}, a_{1:K}, b_{1:K}) = p(\varepsilon_{L}, \varepsilon_{R}) p(\theta)  \nonumber  \\ 
& \cdot \prod_{k = 1}^{K} \frac{\mathbb{I}((\alpha_{k}, \beta_{k}) \in \mR_{\eta_{\nu_{k}}}(\varepsilon_{L}, \varepsilon_{R})}{|\mR_{\eta_{\nu_{k}}}(\varepsilon_{L}, \varepsilon_{R})|}  g_{\theta}(a_{k}, b_{k} | \alpha_{k}, \beta_{k}) \label{eq: joint pdf}
\end{align}
This joint distribution can be used to generate samples via MCMC from the posterior distribution,
\[
p (\varepsilon_{L}, \varepsilon_{R}, \theta, \alpha_{1:K}, \beta_{1:K} | a_{1:K}, b_{1:K}),
\]
which is proportional to the joint distribution in \eqref{eq: joint pdf}. We present Algorithm \ref{alg: MCMC-DP-Est} in Appendix \ref{sec: MCMC algorithm for Bayesian Estimation of BMP}, which, being a variant of MHAAR \citep{andrieu_et_al_2020}, is a suitable MCMC algorithm for latent variable models such as the one in \eqref{eq: joint pdf}.

\begin{remark}
It is straightforward to modify the methodology for the other privacy definitions. For example, for BMP-L (resp.\ BMP-R), it suffices to fix $\varepsilon_{R} = \infty$ (resp.\ $\varepsilon_{L} = \infty$). For MP, one can enforce $\varepsilon_{L} = \varepsilon_{R}$ in the joint probability model and run Algorithm \ref{alg: MCMC-DP-Est} with flat priors, i.e., $\gamma_{k}(1) = 0.5$ for all $k$.
\end{remark}

\subsection{Designing MIAs and measuring their performance} \label{sec: Designing MIAs and measuring their performance}
Algorithm \ref{alg: node-MIA-v1} includes a node-level MIA for algorithm $\mA = (\mS, \Phi)$ and assesses its performance. First, we sample $J$ training graphs under $\mS$ and train GNNs on each graph, obtaining shadow models $\Phi_{1:J}$. Then, for each node $\nu_{k} \in \mathcal{T}$ and model $\Phi_j$, we form a \emph{challenge} $(\nu_{k}, \Phi_{j})$. For this challenge, all the shadow models \emph{except} $\Phi_{j}$ are used to learn the absence ($\mH_{0}$) and presence ($\mH_{1}$) of $\nu_{k}$. A decision is made for the challenge $(\nu_{k}, \Phi_{j})$ using the learned hypotheses (see Algorithm \ref{alg: MIA with posterior logit}), and the (possible) errors are $A_{k}, B_{k}$ counted. Hence, the numbers of challenges for which $H_{0}, H_{1}$ are true are 
\[
n_{0, \nu}:=\sum_{j = 1}^{J} \mathbb{I}(\nu \notin S_{j}), \quad n_{1, \nu}:= J - n_{0, \nu},
\]
respectively. This procedure requires $\mathcal{O}(KJ)$ computations.


\begin{algorithm}
\caption{Node-MIA}
\label{alg: node-MIA-v1}
\begin{algorithmic}[1]
\Require{Graph $\mG = (\mV, \mE)$, algorithm $\mA = (\mS, \Phi)$, target set $\mathcal{T} = \{\nu_{1},\ldots, \nu_{K}\}$, number of shadow models $J$}
\Ensure{True counts $n_{0, 1:K}$, $n_{1, 1:K}$; errors $A_{1:K}$, $B_{1:K}$}
\For{$j = 1, \ldots, J$}
\State $S_{j} \sim  \mS(\mG)$  {\Comment{\footnotesize Sample training graphs}}
\State $\Phi_{j} = \Phi(\mG_{S_{j}}, \bX_{S_{j}}, \bY_{S_{j}}).$ {\Comment{\footnotesize Train $J$ GNNs}}
\EndFor
\For{$k = 1, \ldots, K$}
\State Set $\nu = \nu_{k}$ {\Comment{\footnotesize Target node}}
\State $\hat{\gamma}_{\nu}(1) = \frac{1}{J}\sum_{j = 1}^{J} \mathbb{I}(\nu \in S_{j})$. {\Comment{\footnotesize Estimate prior}}
\State Test every shadow model for the membership of $\nu$:
\State Set $A_{k} = B_{k} = 0$, $n_{0, k} = n_{0, k} = 0$
\For{$j \in \{1, \ldots, J\}$}
\State $m = \mathbb{I}(\nu \in S_{j})$. {\Comment{\footnotesize True membership}}
\State $n_{m, k} \leftarrow n_{m, k} + 1$ {\Comment{\footnotesize True count}}
\State Construct the sets to learn $H_{0}$ and $H_{1}$:
\State  $\mH_{0} =  \{\Phi_{i}: i \in \{1, \ldots, J\}-
\{j\}, \nu \notin S_{i}  \}$ 
\State $\mH_{1} = \{\Phi_{i}: i \in \{1, \ldots, J\}-\{j\}, \nu \in S_{i} \}$,
\State $\hat{m} \gets \textsc{Decision}(\mathcal{I}_\nu, y_\nu,\mH_{0},\mH_{1}, \hat\gamma_{\nu}(1))$
\State $A_{k} \leftarrow A_{k} + \mathbb{I}( \hat{m} \neq m = 0)$ {\Comment{\footnotesize type-I error}}
\State $B_{k} \leftarrow B_{k} + \mathbb{I}(\hat{m} \neq m = 1)${\Comment{\footnotesize type-II error}}
\EndFor
\EndFor
\end{algorithmic}
\end{algorithm}

\begin{algorithm}
\caption{\textsc{Decision}($\mathcal{I}_\nu, y_\nu, \Phi, \mH_{0}, \mH_{1},  \gamma_{\nu}(1) $)}
\label{alg: MIA with posterior logit}
\begin{algorithmic}[1]
\Require{Target node query $\mathcal{I}_\nu$ and its true label $y_\nu$, target model $\Phi$; shadow models w. \& w.o membership of node $\nu$, $\mH_{0} = \{ \Phi_{0, 1}, \ldots, \Phi_{0, n_{0}}\}$, $\mH_{1} = \{ \Phi_{1, 1}, \ldots, \Phi_{1, n_{1}}\}$, prior membership probability $\eta_{\nu}(1)$}
\Ensure{Decision $\hat{m}$}
\For{$j = 0, 1$}
\For{$i = 1, \ldots, n_{j}$}
\State $z_{j, i} = [\Phi_{j, i}(\mathcal{I}_\nu)]_{y_{\nu}} - \log \sum_{y \neq y_{\nu}} e^{[\Phi_{j, i}(\mathcal{I}_\nu)]_{y}}$ 
\EndFor
\State Fit a normal distribution for $\mH_{j}$:
\[
\mu_{j} \gets \mathrm{mean}(z_{j,1:n_{j}}\big), \quad \sigma_{j} \gets \mathrm{std}\left(z_{j, 1:n_{j}}\right)
\]
       
\EndFor

\State Calculate the test value:
\[
z \gets [\Phi(\nu)]_{y_{\nu}} - \log \sum_{y \neq y_{\nu}} e^{[\Phi(\nu)]_{y}}.
\]

\State \Return Decision 
\[
\hat{m} = \mathbb{I} \left[ \frac{\mathcal{N}(z; \mu_{1}, \sigma_{1}^{2})}{\mathcal{N}(z; \mu_{0}, \sigma_{0}^{2})} \frac{\gamma_{\nu}(1)}{1-\gamma_{\nu}(1)}> 1 \right]
\]

\end{algorithmic}
\end{algorithm}

One can also decouple the models used to form challenges from those used to fit the member/non-member hypotheses; this alternative is shown in Algorithm~\ref{alg: node-MIA-v2} in Appendix~\ref{apndx: MIA alternative}.

\section{Experiments}

\begin{table*}[ht]
\caption{Hierarchical results table (p5 / p50 / p95 per cell) when \cora is employed with $25\%$ nodes in train set.}
\label{tbl:Cora_25}
\centering
\small
\setlength{\tabcolsep}{3pt}
\renewcommand{\arraystretch}{1.2}

\begin{tabular}{lll|cccc|cccc}
\hline
\multirow{3}{*}{\textbf{Sampling}} & \multirow{3}{*}{\textbf{Decision}} & \multirow{3}{*}{\textbf{Output}}
 & \multicolumn{4}{c|}{\textbf{\gcn}} & \multicolumn{4}{c}{\textbf{\gat}} \\

 &  &  & \multicolumn{2}{c}{\textbf{full edge}} & \multicolumn{2}{c|}{\textbf{no edges}} & \multicolumn{2}{c}{\textbf{full edge}} & \multicolumn{2}{c}{\textbf{no edges}} \\

 &  &  & \textbf{BMP-R} & \textbf{MP} & \textbf{BMP-R} & \textbf{MP} & \textbf{BMP-R} & \textbf{MP} & \textbf{BMP-R} & \textbf{MP} \\
\hline
\multirow{4}{*}{\textbf{random}} & \multirow{2}{*}{\textbf{weak}} & \textbf{clean}  & 2.2/2.3/2.4 & 3.5/3.5/3.8 & 6.2/7.7/10 & 7/8.3/11 & 6.2/7.4/9.9 & 7.3/8.6/11 & 1.8/1.9/2 & 3.7/3.9/4 \\
 &  & \textbf{noisy}  & 1.6/1.6/1.6 & 2.9/2.9/2.9 & 3.9/4.3/5.2 & 4.8/5.1/5.5 & 3.8/4.2/4.9 & 4.8/5/5.3 & 1.5/1.5/1.5 & 3/3.1/3.1 \\
 & \multirow{2}{*}{\textbf{strong}} & \textbf{clean}  & 4.2/5.4/8.2 & 4.2/4.5/4.8 & 6.6/8.1/11 & 7.4/8.7/11 & 6.9/8.3/11 & 7.9/9.1/12 & 3.1/3.5/4.2 & 3.8/3.9/3.9 \\
 &  & \textbf{noisy}  & 2.7/3.2/5.9 & 3.1/3.2/3.3 & 4.8/5.6/7.8 & 5.4/5.7/6.3 & 4.7/5.3/7.2 & 5.4/5.7/6.1 & 2.1/2.2/2.4 & 3/3/3.1 \\
\hline
\multirow{4}{*}{\textbf{snowball}} & \multirow{2}{*}{\textbf{weak}} & \textbf{clean}  & 4.4/4.8/5.7 & 9.7/11/13 & 6.1/7.4/9.8 & 9.8/11/13 & 4.8/5.3/6.5 & 9.4/11/13 & 5.1/5.8/8 & 9.6/11/13 \\
 &  & \textbf{noisy}  & 4/4.4/4.9 & 9.1/11/13 & 5.2/6.2/8.4 & 9.2/11/13 & 4.4/5/6.2 & 9.5/11/13 & 5/5.8/7.8 & 9.3/11/13 \\
 & \multirow{2}{*}{\textbf{strong}} & \textbf{clean}  & 5.1/6.1/8.8 & 9.9/11/13 & 5.5/6.7/9.2 & 10/11/13 & 5.2/6.3/8.8 & 9.6/11/13 & 5.6/6.9/9.5 & 9.8/11/13 \\
 &  & \textbf{noisy}  & 4.2/4.7/5.7 & 9.4/11/13 & 4.7/5.4/7.4 & 9.6/11/13 & 4.6/5.4/7.6 & 9.5/11/13 & 5/5.8/7.9 & 9.5/11/13 \\
\hline
\end{tabular}
\end{table*}

We evaluate privacy parameters across a systematically structured experimental space varying (i) dataset (\cora, \cseer), (ii) model architecture (\gcn, \gat), (iii) training fraction (25\%, 50\%), and (iv) sampling mechanism (random vs. snowball). Details of the datasets and the GNN architecture are provided in Appendix \ref{sec:datasetsandmodels}. 

\paragraph{Sampling strategies} We employ two sampling strategies to construct training graphs. In random sampling, nodes are selected independently and uniformly at random, and the induced subgraph is used for training. In snowball sampling, we start from a random seed node and iteratively add up to a fixed number of randomly chosen neighbors until the target size is reached, yielding a structurally connected training subgraph.

\paragraph{Attack Instantiation} For each configuration, we instantiate the attack using Algorithm~\ref{alg: node-MIA-v1} with
$J=1000$ shadow models, each mimicking the target mechanism 
$\mA$. For every node 
$v\in \mathcal{V}$, prior membership probabilities are estimated as the empirical frequency with which 
$v$ appears across the $N$ shadow training sets. These empirical frequencies are treated as the true priors when estimating privacy parameters. The adversary is assumed to know the true edge structure, i.e., the whole graph is used when querying the models (referred to as \emph{full edges} case in our results).

\paragraph{Attack variants} We additionally study several attack variants. First, we consider a \emph{weak test} in which the decision function relies solely on prediction likelihoods, without incorporating prior probabilities. Second, we inject noise into the last layer of the trained model. Third, we evaluate a \emph{no-edge} inference setting, where all graph edges are removed when computing prediction probabilities for the attack. Importantly, these modifications do not alter the underlying trained models; they only affect the information available to the adversary at inference time.

 

We exhaustively estimate both BMP-R and MP for all combinations of experiments. In addition, we estimate the most general BMP for a selected subset of experiments. The implementation details, including the selection of MCMC parameters, are given in Appendix \ref{sec: Implementation details}.

\paragraph{Results} In Table~\ref{tbl:Cora_25}, we present the 5th, 50th (median), and 95th percentile (p5 / p50 / p95) estimates of the privacy parameters for MP and BMP-R on the \cora\ dataset, when $25\%$ of the nodes are used for training. Tables \ref{tbl:Cora_50}, \ref{tbl:CiteSeer_25}, \ref{tbl:CiteSeer_50} in Appendix \ref{sec: Detailed Experimental Results} show the same results when \cora with 50\% sampling rate and \cseer with 25\% and 50\% sampling rates, respectively. Moreover, Figure \ref{fig: marginal effects} shows the marginal differences between the posterior means of privacy parameters obtained by choosing one option vs the other for the experiment configuration parameters, as indicated on the $x$-axes. 

\begin{figure}
\begin{subfigure}[t]{\linewidth} 
\centering
\includegraphics[scale = 0.95]{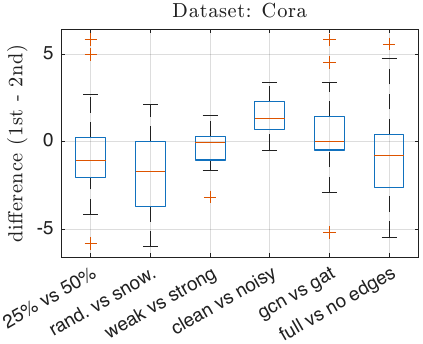}
\caption{\cora}
\end{subfigure}
\begin{subfigure}[t]{\linewidth} 
\centering
\includegraphics[scale = 0.95]{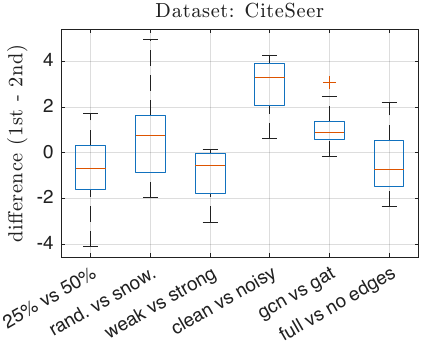}
\caption{\cseer}
\end{subfigure}
\caption{Marginal effects of the experiment parameters on $\epsilon_{R}$-BMP-R}
\label{fig: marginal effects}
\end{figure}

From Figure~\ref{fig: marginal effects}, we conclude that, for \cora and \cseer, \emph{sampling rate} (25\% vs.\ 50\%) and \emph{noise injection} (clean vs.\ noisy) are the dominant drivers of privacy-risk differences. Also, all factors affect privacy in the same direction for both datasets, except for \emph{sampling type} (random vs snowball). While snowball sampling seems less private for \cora, it seems slightly more private for \cseer. Finally, for \cora, a closer inspection of Table \ref{tbl:Cora_25} across random vs snowball sampling shows that the effects of edge availability and model architecture are more discriminative under random sampling. The detailed results also reveal that in almost every configuration (except under random sampling with 50\% rate), the BMP-R privacy parameter estimates are significantly smaller than those of MP. This suggests that the nodes with high likelihood ratios (hence high $\varepsilon_{\text{MP}}$) appear to have small membership priors, dampening the posterior membership probability and resulting in a small $\varepsilon_{R}$ relative to $\varepsilon_{\text{MP}}$. Figure \ref{fig: priors and bounds} illustrates this phenomenon for an example case.

\begin{figure}
\centering
\includegraphics[scale = 1.1]{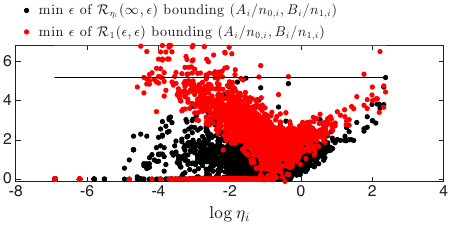}
\caption{Each dot corresponds to a target node in $\mathcal{T}$. Those nodes whose empirical error rates imply larger $\varepsilon$ tend to have smaller membership priors. As a result, $\varepsilon_{\text{MP}}$ estimated to be larger than $\varepsilon_{R}$ of BMP-R. (The experiment: \cora, 25\% sampling rate, \gcn, clean outputs, full edges, strong attack.)}
\label{fig: priors and bounds}
\end{figure}
\begin{figure}[h!]
\centerline{
\includegraphics[scale = 1.2]{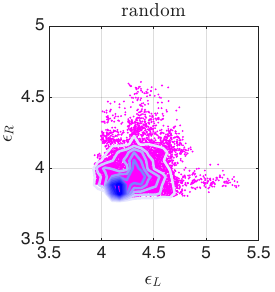}
}
\centerline{
\includegraphics[scale = 1.2]{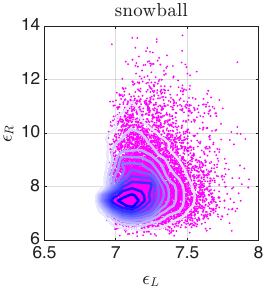}
}
\caption{Post.\ samples of \((\varepsilon_R,\varepsilon_L)\) of BMP for random vs.\ snowball samp.\ on \cora (50\% samp., \gcn, clean outputs).}
\label{fig: posteriorsamples}
\end{figure}
Figure \ref{fig: posteriorsamples} exemplifies joint estimation of $(\varepsilon_R,\varepsilon_L)$ of BMP and compares the posterior samples for $(\varepsilon_R,\varepsilon_L)$ for random and snowball sampling for \cora for 50\% split size. We observe that snowball sampling shifts the posterior mass toward substantially larger values of both privacy parameters, indicating higher privacy leakage than random sampling. We also observe that, in snowball sampling, the posterior estimates of $\varepsilon_{L}$ and $\varepsilon_{R}$ significantly differ, suggesting that the posterior membership and non-membership probabilities have different upper bounds, thereby justifying the presence of two separate parameters in the inequality defining BMP.

%
%


\section{Conclusion}
We introduced Bayesian Membership Privacy (BMP), a sampling-aware framework for analyzing node-level membership leakage in GNNs within a Bayesian framework. By incorporating node-dependent priors and casting membership inference as Bayesian hypothesis testing, BMP captures heterogeneous privacy risk that uniform metrics may miss. Our designed MIA and the accompanying MCMC posterior sampling method enable computation of uncertainty-aware estimates of privacy parameters of an arbitrary GNN. 

In our experiments, we estimate BMP parameters numerically using posterior sampling based on the hierarchical model developed in Section \ref{sec: Joint probability model and Bayesian estimation via MCMC}. An interesting line of future work is developing mechanisms that provide BMP with analytical guarantees.
\section*{acknowledgements} MK is supported by the Dutch Research Council (NWO) through the Vidi project PriXAI, file number 
VI.Vidi.243.224 of the research programme Vidi ENW under the grant https://doi.org/10.61686/FATII77836.

\bibliographystyle{apalike}
\bibliography{references}


\newpage

\onecolumn


\appendices

\section{Related Bayesian definitions} \label{sec: Related definitions} 

Our privacy definitions are closest to the posterior membership privacy PMP definitions of \citet{li2013membership}. In both our work and \citet{li2013membership}, the population (the whole dataset) is assumed to be known, and the sensitive unknown information is the membership of the elements in the population. Also, both our definitions consider the posterior membership probabilities, and we note that some one-to-one conversions are possible in some special cases. However, there are a few differences between the approaches: While PMP focuses on the difference between the posterior and the prior as a metric of informativeness of the algorithm, BMP focuses directly on the posterior, reflecting our view that the sampling phase is an integral part of the algorithm (hence the notation $\mA = (\mS, \Phi)$.). Also, while \citet{li2013membership} introduce ``positive PMP'' and ``negative PMP'' (in parallel to BMP-R and BMP-L) corresponding to membership and non-membership posterior probabilities, our BMP Definition \ref{def: BMP} unifies the two cases. 
Finally, while we consider a single sampling distribution $\mS$, attached to the algorithm $\mathcal{A}$, \citet{li2013membership} assumes a set of sampling distributions for which the defining inequality holds.

Differential identifiability of \citet{Lee_and_Clifton_2012} is also related to our work in that their definition corresponds to the $\varepsilon$-BMDP-R.


There also exist related definitions that employ the Bayesian paradigm, designed for general data structures, yet applicable to membership information as a special case. For example, \citet{Yang_et_al_2015} propose a Bayesian DP (BDP) definition for correlated sensitive data. BDP concerns the likelihood ratio involving an unknown component of the sensitive data, conditional on a known part of the sensitive data. The MP in Definition \ref{def: MP} corresponds to the BDP of \citet{Yang_et_al_2015} when the membership information is viewed as correlated data, so that the unknown element is a node, and no membership information is assumed known. Unlike our approach, the term ``Bayesian'' is used not because of the involvement of the posterior distribution of the unknown element in question; it is rather because of assigning a probability distribution on the sensitive data so that marginal likelihoods can be computed via marginalization. Another definition, with the same name BDP, is proposed in \citet{Triastcyn_and_Faltings_2020}. The definition is reminiscent of the classical DP, with the difference that the ratio of likelihoods is treated as a random variable by employing the prior of the differing element between the datasets.

Our definition also bears similarities to PAC Membership Privacy \citep[Def.~5]{Xiao_and_Devadas_2023}: both are non-differential notions and assume that the adversary knows the population and the sampling process, but not the realized sampled training set. A key distinction, however, is that BMP explicitly supports non-uniform, node-dependent membership priors induced by graph sampling, making it better suited to graph settings with heterogeneous inclusion probabilities.

\section{Missing Proofs} \label{sec: Missing Proofs}

\begin{proof}[Proof of Proposition \ref{prop: lower bound on BMP}]
We will only prove the first inequality; the second one follows similarly. Fix $i \in \{1, \ldots, N\}$. We have $\mathbb{E}(\pi_{i}(1 | W)) = \gamma_{i}$, where the expectation is with respect to the marginal distribution $W$ (with density $f_{0}(\cdot)$). This means that there is value $w \in \mathcal{W}$ with $f_{0}(w) > 0$ such that $\pi_{i}(1 | w) \geq \gamma_{i}(1)$, which implies that $ \pi_{i}(1 | w)/\pi_{i}(0 | w) \geq \gamma_{i}(1)/\gamma_{i}(0)$. 
Since $\mA$ is $\varepsilon$-BMP, we have $e^{\varepsilon_{R}} \geq \pi_{i}(1 | w)/\pi_{i}(0 | w) \geq \gamma_{i}(1)/\gamma_{i}(0)$. Since the above holds for any $i \in \{1, \ldots, N\}$, we have the claimed lower bound for $\varepsilon_{R}$.
\end{proof}

For the proofs of Proposition \ref{prop: post-processing} and Theorems \ref{thm: R for BMP} and \ref{thm: expected cost and wrong decision prob}, note the useful fact that the definitions in Section \ref{sec: Definitions for membership privacy} can be written in terms of output sets with positive probability, i.e., 
\[
e^{-\varepsilon_{L}} \leq \frac{\pi_{i}(1 | w)}{\pi_{i}(0 | w)} \leq e^{\varepsilon_{R}} \Leftrightarrow e^{-\varepsilon_{L}} \leq \frac{\pi_{i}(1 | O)}{\pi_{i}(0 | O)} \leq e^{\varepsilon_{R}}
\]
for any $O \subseteq \mathcal{W}$ with $\Pr(W \in O) > 0$ the marginal and full conditional posteriors
\begin{align*}
&\pi_{i}(m_{i} | O) := \Pr(M_{i} = m_{i} | W \in O). 
\end{align*}

\begin{proof}[Proof of Proposition \ref{prop: post-processing}]
To distinguish between the posterior membership probabilities conditional on the outputs of $\mA$ and $\mathcal{B}$, we denote those probabilities as $\pi^{\mA}$ and $\pi^{\mB}$, respectively. Let $W = \Phi(\mG_{S}, X_{S})$ with $S \sim \mS(\mG)$ and let $Z = \varphi(\Phi(\mG_{S}, X_{S}))$ be the output of $\mathcal{B}$. Fix $i \in \{1, \ldots, N\}$. For any $O_{z} \subseteq \mathcal{Z}$, we have $\{ Z \in O_{z} \} = \{ W \in O_{y} \}$, where $O_{w} = \varphi^{-1}(O_{z}) = \{ w \in \mathcal{W}: \varphi(w) \in O_{z} \} \subseteq \mathcal{W}$. Hence, $\pi^{\mB}_{i}(1 | O_{z}) := \Pr(M_{i} = 1 | Z \in O) = \Pr(M_{i} = 1 | W \in O_{w}) = \pi^{\mA}_{i}(1 | O_{w})$. Since $\mA$ is $(\varepsilon_{L}, \varepsilon_{R})$,
\[
e^{-\varepsilon_{L}} \leq \pi^{\mA}_{i}(1 | O_{w})/\pi^{\mA}_{i}(0 | O_{w}) = \pi^{\mB}_{i}(1 | O_{z})/\pi^{\mB}_{i}(0 | O_{z}) \leq e^{\varepsilon_{R}}.
\]
 Hence, the claim is proven.
\end{proof}

\begin{proof}[Proof of Proposition \ref{prop: MP to MPP}]
Assume $\mA$ is $\varepsilon$-MP. Then for any $w \in \mathcal{W}$, and for all $i \in \{1, \ldots, N\}$, we have
\[
e^{-\varepsilon} \eta_{i} \leq \frac{\pi_{i}(1 | w)}{\pi_{i}(0 | w)} \leq \eta_{i} e^{\varepsilon}
\]
or
\[
e^{-(\varepsilon - \log \eta_{i})} \leq \frac{\pi_{i}(1 | w)}{\pi_{i}(0 | w)} \leq e^{\varepsilon + \log \eta_{i}}
\]
A uniform bound can be developed for the ratio $\frac{\pi_{i}(1 | w)}{\pi_{i}(0 | w)}$ as
\[
e^{-(\varepsilon - \log \eta_{\min})} \leq \frac{\pi_{i}(1 | w)}{\pi_{i}(0 | w)} \leq e^{\varepsilon + \log \eta_{\max}}
\]
Therefore, $\mathcal{A}$ is $(\varepsilon_{L}, \varepsilon_{R})$-BMP, with
\[
\varepsilon_{L} = \varepsilon -  \log \eta_{\min}, \quad \varepsilon_{R} = \varepsilon +  \log \eta_{\max}.
\]

For the other direction, assume $\mathcal{A}$ is $(\varepsilon_{L}, \varepsilon_{R})$-BMP. Then
\[
e^{-\varepsilon_{L}} \leq \frac{\pi_{i}(1 | w)}{\pi_{i}(0 | w)} \leq e^{\varepsilon_{R}},
\]
or, 
\[
e^{-(\varepsilon_{L} + \log \eta_{i})} \leq \frac{f_{i}(w | 1)}{f_{i}(w | 0)} \leq e^{\varepsilon_{R} - \log \eta_{i}},
\]
A uniform bound can be developed for the ratio $\frac{f_{i}(w | 1)}{f_{i}(w | 0)}$ as
\[
e^{-(\varepsilon_{L} + \log \eta_{\max})} \leq \frac{f_{i}(w | 1)}{f_{i}(w | 0)} \leq e^{\varepsilon_{R} - \log \eta_{\min}}
\]
Considering the maximum of the two exponents, we conclude that $\mathcal{A}$ is $\varepsilon$-MP where 
\[
\varepsilon = \max\left\{\varepsilon_{L} + \log \eta_{\max},  \varepsilon_{R} - \log \eta_{\min} \right\}.
\]
\end{proof}

\begin{proof}[Proof of Proposition \ref{prop: Composition with classical MP}]
For any given $w_{0:T} \in \mathcal{W}^{T+1}$, we have 
\[
\frac{\pi_{i}(m_{i} = 1 | W_{0:T})}{\pi_{i}(m_{i} = 0 | W_{0:T})} = \frac{\gamma_{i}(1) f^{(0)}_{i}(w_{0} | 1) f^{(1)}_{i}(w_{1} | 1) \ldots f^{(T)}_{i}(w_{T} | 1)}  {\gamma_{i}(0) f^{(0)}_{i}(w_{0} | 0) f^{(1)}_{i}(w_{1} | 0) \ldots f^{(T)}_{i}(w_{T} | 0)} = \frac{\pi_{i}(1 | w_{0})}{\pi_{i}(0 | w_{0})} \frac{ f^{(1)}_{i}(w_{1} | 1) \ldots f^{(T)}_{i}(w_{T} | 1)}  { f^{(1)}_{i}(w_{1} | 0) \ldots f^{(T)}_{i}(w_{T} | 0)} 
\]
Since $\mA_{t}$ is $\varepsilon_{t}$-MP, we have
 \begin{equation} \label{eq: MP ineq comp}
\exp\left\{-\sum_{t = 1}^{T} \varepsilon_{t} \right\} \leq \frac{ f^{(1)}_{i}(w_{1} | 1) \ldots f^{(T)}_{i}(w_{T} | 1)}  { f^{(1)}_{i}(w_{1} | 0) \ldots f^{(T)}_{i}(w_{T} | 0)} \leq \exp\left\{\sum_{t = 1}^{T} \varepsilon_{t} \right\}
 \end{equation}
 Also, since $\mA_{0}$ is $(\varepsilon_{L}, \varepsilon_{R})$-BMP, we have
 \begin{equation} \label{eq: BMP ineq comp}
e^{-\varepsilon_{L}} \leq \frac{\pi_{i}(1 | w_{0})}{\pi_{i}(0 | w_{0})} \leq e^{\varepsilon_{R}}
 \end{equation}
 Combining \eqref{eq: MP ineq comp} and \eqref{eq: BMP ineq comp}, we have the desired bound. 
 
 Finally, the result would not change if $W_{t}$ depends on $W_{0:t-1}$ for $t = 1, \ldots, T$: In that case, the ratio in \eqref{eq: MP ineq comp} would change to
 \[
 \frac{ f^{(1)}_{i}(w_{1} | w_{0}, 1)}{f^{(1)}_{i}(w_{1} | w_{0}, 0) } \ldots \frac{f^{(T)}_{i}(w_{T} | w_{0:T-1}, 1)}{ f^{(T)}_{i}(w_{T} | w_{0:T-1}, 0)}
 \]
 which can be bounded in the same way because each $\mathcal{A}_{t}$ is $\varepsilon_{t}$-MP for any $W_{0:t-1} = w_{0:t-1}$.
\end{proof}

\begin{proof}[Proof of Proposition \ref{prop: Composition of BMPs}]
Let the membership information of the $t$th sampling be $M^{(t)}$. Then, given $W_{1:T} = w_{1:T}$
\[
\Pr\!\left(\exists t \in (1,\ldots,T): M_i^{(t)} = 1 \mid W_{t:T} = w_{1:T}\right)
\le
1 - \prod_{t=1}^{T} (1-\kappa_t)
\]
where $\kappa_t = \frac{e^{\varepsilon_R^{(t)}}}{1 + e^{\varepsilon_R^{(t)}}}$. Also, we can prove that
\[
\Pr\!\left(\forall t \in (1,\ldots,T): M_i^{(t)} = 0 \mid W_{t:T} = w_{1:T} \right) \le \prod_{t=1}^{T} \mu_t
\]
where $\mu_t = \frac{e^{\varepsilon_L^{(t)}}}{1 + e^{\varepsilon_L^{(t)}}}$. Combining both inequalities, we can show that the composition is
\((\varepsilon_L,\varepsilon_R)\)-BMP, where
\[
\varepsilon_L=\sum_{t=1}^{T} \log \mu_t-\log\!\left[1 - \prod_{t=1}^{T} \mu_t\right], \quad \varepsilon_R =\log\!\left[1 - \prod_{t=1}^{T} (1-\kappa_t)\right]-\sum_{t=1}^{T} \log(1-\kappa_t).
\]
So, the claim is proven.
\end{proof}

\begin{proof}[Proof of Theorem \ref{thm: R for BMP}]
For any $i \in \{1, \ldots, N\}$, $m_{i} \in \{0, 1\}$, and $O \subseteq \mathcal{W}$, let $F_{i}(O | m_{i}) := \Pr(W \in O | M_{i} = m_{i})$. Any decision rule based on the output of $\mA$ can be written in terms of a critical region $C \subseteq \mathcal{W}$ as $D(W) = \mathbb{I}(W \in C)$ for a measurable set $C$. Definition \ref{def: BMP} guarantees that $\gamma_{i}(0) F_{i}(C | 0) \leq e^{\varepsilon_{L}} \gamma_{i}(1) F_{i}(C | 1)$. Using $\alpha_{i} = F_{i}(C | 0)$, $\beta_{i} =1 - F_{i}(C | 1)$, we have
\[
\eta_{i}^{-1} e^{-\varepsilon_{L}} \alpha_{i} + \beta_{i} \leq 1
\]
Definition \ref{def: BMP} also implies $\gamma_{i}(1) F_{i}(C | 1) \leq e^{\varepsilon_{R}} \gamma_{i}(0)  F_{i}(C | 0)$. Using the same substitutions for $\alpha_{i}$ and $\beta_{i}$, we get
\[
\eta_{i}^{-1} e^{\varepsilon_{R}} \alpha_{i} + \beta_{i} \geq 1.
\]
Next, applying Definition \ref{def: BMP} with the critical region being $C^{c}$ so that $\alpha_{i} = F_{i}(C^{c} | 0)$, $\beta_{\nu} = F_{i}(C | 1)$. We have $\gamma_{i}(0) F_{i}(C^{c} | 0) \leq e^{\varepsilon_{L}} \gamma_{i}(1) F_{i}(C^{c} | 1)$. Using the substitutions for $\alpha_{i}$ and $\beta_{i}$, we get
\[
\alpha_{i} + \eta_{i} e^{\varepsilon_{L}} \beta_{i} \geq 1
\]
Finally, the right-hand inequality for $C^{c}$ reads $\gamma_{i}(1) F_{i}(C^{c} | 1) \leq e^{\varepsilon_{R}} \gamma_{i}(0) F_{i}(C^{c} | 0)$, which yields
\[
\eta_{i} e^{-\varepsilon_{R}} \beta_{i} + \alpha_{i} \leq 1.
\]
Hence, all the inequalities in $\mR_{i}(\varepsilon_{L}, \varepsilon_{R})$ are satisfied. For the reverse implication, it suffices to go backwards from each equality to arrive at the inequality in the Definition \ref{def: BMP}.
\end{proof}

\begin{proof}[Proof of Theorem \ref{thm: expected cost and wrong decision prob}]
For a given decision rule $D$, the expected cost is
\begin{align*}
\mathbb{E}[\text{Cost}] &=c_{0}  \int \mathbb{I}_{1}(D(w)) \gamma_{i}(0)  f_{i}(w| 0) \md w +  c_{1} \int \mathbb{I}_{0}(D(w)) \gamma_{i}(1) F_{i}(w| 1) \md w  \\
& =  \int [ c_{0} \pi_{i}(0 | w) \mathbb{I}_{1}(D(w)) + c_{1} \pi_{i}(1 | w) \mathbb{I}_{0}(D(w))  ]  f_{0}(\md w) \\
& = \mathbb{E}[c_{0} \pi_{i}(0 | W)  \mathbb{I}_{1}(D(W))  + c_{1} \pi_{i}(1 | W) \mathbb{I}_{0}(D(W))  ]
\end{align*}
This is minimized when 
\[
D(w) = \begin{cases} 1 & c_{0} \pi_{i}(0| w) \leq c_{1} \pi_{i}(1| w) \\
0 & c_{0} \pi_{i}(0 | w) > c_{1} \pi_{i}(1| w)
\end{cases}, \quad w \in \mathcal{W}
\]
and the minimum expected cost $E_{min}$ is 
\[
E_{min} = \mathbb{E}[\min \{ c_{0} \pi_{i}(0|W), c_{1} \pi_{i}(1|W) \}].
\]
By Definition \ref{def: BMP}, $e^{-\varepsilon_{L}} \leq \pi_{i}(1|W)/\pi_{i}(0|W) \leq e^{\varepsilon_{R}}$, which yields
\[
\min \{ c_{0} \pi_{i}(0 | W), c_{1} \pi_{i}(1|W) \} \geq \min\left\{ c_{0} \frac{1}{1 + e^{\varepsilon_{R}}}, c_{1}  \frac{1}{1 + e^{\varepsilon_{L}}} \right\}
\]
On the other hand, since $\pi_{i}(0|W)/\pi_{i}(1|W) \geq e^{-\varepsilon_{R}}$, we have 
\begin{align*}
E_{\min} &= \mathbb{E}\left[ \pi_{i}(1|W) \min \left\{ c_{0}\frac{\pi_{i}(0|W)}{\pi_{i}(1|W)}, c_{1} \right\} \right]\\
& \geq \mathbb{E}\left[ \pi_{i}(1|W) \min\{ c_{0} e^{-\varepsilon_{R}}, c_{1} \}\right] \\
& \geq \gamma_{i}(1) \min\{ c_{0} e^{-\varepsilon_{R}}, c_{1} \}
\end{align*}
where in the last line we have used the fact that $\mathbb{E}[\pi_{i}(1|W)] = \gamma_{i}(1)$. 
By following similar steps but using $\pi_{i}(0|W)/\pi_{i}(1|W) \geq e^{-\varepsilon_{L}}$, we also have 
\begin{align*}
E_{\min} &\geq \gamma_{i}(0) \min\{ c_{1} e^{-\varepsilon_{L}}, c_{0} \}
\end{align*}
The three lower bounds can be combined to yield \eqref{eq: lower bound for cost}. In particular, the probability of a wrong decision is obtained by $\mathbb{E}(\text{Cost})$ with $c_{1} = c_{2} = 1$, yielding the lower bound \eqref{eq: lower bound for wrong decision probability}.
\end{proof}

\section{MCMC algorithm for Bayesian Estimation of BMP} \label{sec: MCMC algorithm for Bayesian Estimation of BMP}

\begin{algorithm}
\begin{algorithmic}
\caption{MCMC for posterior sampling for $(\varepsilon_{L}, \varepsilon_{R}, \theta)$}
\label{alg: MCMC-DP-Est}
\Require{Numbers of true non-membership and membership instances $n_{0, 1:K}, n_{1, 1:K}$; error counts $A_{1:K}$, $B_{1:K}$; \\ 
prior distribution for BMP parameters $p(\varepsilon_{L}, \varepsilon_{R})$; prior distribution $p(\theta)$ on $\theta$; \\
proposal distributions $q(\varepsilon_{L}', \varepsilon_{R}' | \varepsilon_{L}, \varepsilon_{R})$, $q(\theta' | \theta)$; proposal distribution for $(\alpha_{k}, \beta_{k})$, $q_{k}(\alpha, \beta)$, $k = 1, \ldots, K$; \\
$M$: number of auxiliary variables for each node $k \in \{1, \ldots, K\}$; \\
$T$: Number of MCMC iterations; \\
Initial values: $\varepsilon_{L}, \varepsilon_{R}, \theta$.}
\Ensure{$\{ \varepsilon_{L}^{(t)}, \varepsilon_{R}^{(t)}, s^{(t)}, \theta^{(t)}; t = 1, \ldots, T \}$: MCMC samples}
\For{$t = 1, \ldots, T$}
\State Draw the proposal $(\varepsilon_{L}', \varepsilon_{R}') \sim q(\varepsilon_{L}', \varepsilon_{R}' | \varepsilon_{L}, \varepsilon_{R})$, and $\theta' \sim q(\theta' | \theta)$.

\For{$k = 1:K$}
\State Set $(\alpha_{k}^{(1)}, \beta_{k}^{(1)}) = (\alpha_{k}, \beta_{k})$.
\State Sample $\alpha_{k}^{(i)}, \beta_{k}^{(i)} \overset{\text{iid}}{\sim} q_{j}(\alpha, \beta)$ for $i = 2, \ldots, M$.
\State Calculate the weights
\begin{align*}
w_{k}^{(i)} &= \frac{\frac{\mathbb{I}((\alpha_{k}^{(i)}, \beta_{k}^{(i)}) \in \mR_{\eta_{\nu_{k}}}(\varepsilon_{L}, \varepsilon_{R}))}{|\mR_{\eta_{\nu_{k}}}(\varepsilon_{L}, \varepsilon_{R})|}   g_{\theta}(A_{k}, B_{k} | \alpha_{k}^{(i)}, \beta_{k}^{(i)})}{q_{k}(\alpha_{k}^{(i)}, \beta_{k}^{(i)})}, \quad i = 1, \ldots, M\\
w_{k}^{\prime(i)} &= \frac{\frac{\mathbb{I}((\alpha_{k}^{(i)}, \beta_{k}^{(i)}) \in \mR_{\eta_{\nu_{k}}}(\varepsilon'_{L}, \varepsilon'_{R}))}{|\mR_{\eta_{\nu_{k}}}(\varepsilon'_{L}, \varepsilon'_{R})|}   g_{\theta'}(A_{k}, B_{k} | \alpha_{k}^{(i)}, \beta_{k}^{(i)})}{q_{k}(\alpha_{k}^{(i)}, \beta_{k}^{(i)})}, \quad i = 1, \ldots, M
\end{align*}
\EndFor
\State Acceptance probability: 
\[
A = \min \left\{1, \frac{p(\varepsilon_{L}', \varepsilon_{R}')  q(\varepsilon_{L} | \varepsilon_{L}')q(\varepsilon_{R} | \varepsilon_{R}') }{p(\varepsilon_{L}, \varepsilon_{R}) q(\varepsilon_{L}' | \varepsilon_{L})  q(\varepsilon_{R}' | \varepsilon_{R})} \prod_{k = 1}^{K} \frac{\sum_{i = 1}^{M} w_{k}^{\prime(i)}}{ \sum_{i = 1}^{M} w_{k}^{(i)}} \right\}.
\]
\State \textbf{Accept/Reject}: Draw $u \sim \text{Unif}(0, 1)$.
\If{$u \leq A$}
\State Set $(\varepsilon_{L}, \varepsilon_{R}) = (\varepsilon_{L}', \varepsilon_{R}')$, $\theta = \theta'$, and $\bar{w}_{1:K}^{(1:M)} = w_{1:K}^{\prime (1:M)}$.
\Else
\State Keep $\varepsilon_{L}, \varepsilon_{R}, \theta$ and set $\bar{w}_{1:K}^{(1:M)} = w_{1:K}^{(1:M)}$.
\EndIf
\For{$k = 1, \ldots, K$}
\State Sample $j \in \{1, \ldots, M\}$ w.p.\ $\propto \bar{w}_{k}^{(j)}$ and set $(\alpha_{k}, \beta_{k}) = (\alpha_{k}^{(j)}, \beta_{k}^{(j)})$.
\EndFor
\State Store $(\varepsilon_{L}^{(t)}, \varepsilon_{R}^{(t)}) = (\varepsilon_{L}, \varepsilon_{R})$, $\theta^{(t)} = \theta$.
\EndFor
\end{algorithmic}
\end{algorithm}

\begin{remark}
The model in \eqref{eq: joint pdf} neglects any conditional dependency among the error probabilities $(\alpha_{1:K}, \beta_{1:K})$. One way to allow for dependence among the error probabilities is to assume a variable $s \in (0,1)$, with a prior $p(s)$, in the spirit of \citep{Yildirim_et_al_2025}, and let
\[
(\alpha_{k}, \beta_{k}) | \varepsilon, s \sim \textup{Uniform}(\mR_{\eta_{\nu_{k}}}(\varepsilon_{L}, \varepsilon_{R}) - \mR_{(\eta_{\nu_{k}})^{s}}(s\varepsilon_{L}, s\varepsilon_{R}))
\]
Note that in this case $(\alpha_{k}, \beta_{k})$, $k = 1, \ldots, K$, are dependent (after integrating out $s$) conditional on $\varepsilon_{L}, \varepsilon_{R}$. Moreover, the parameter $s$ can be interpreted as a \emph{attack strength} parameter: When $s = 0$, the error probabilities are expected to be anywhere in $\mR_{\eta_{\nu_{j}}(\varepsilon_{L}, \varepsilon_{R}})$ allowing for weakest tests. As $s$ gets larger, the error probabilities are distributed in a narrower region, farther from the $\alpha+\beta = 1$ line where the worst tests lie. 
\end{remark}

\section{Alternative MIA} \label{apndx: MIA alternative}
An alternative MIA is given in Algorithm \ref{alg: node-MIA-v2}. In this alternative version, one decouples the models used to form challenges from those used to fit the member/non-member hypotheses.

\begin{algorithm}
\caption{Node-MIA: Alternative 2}
\label{alg: node-MIA-v2}
\begin{algorithmic}
\State 
\For{$j = 1, \ldots, J^{(1)}$}
\State $S^{(1)}_{j} \sim  \mS(\mG)$,  $\Phi^{(1)}_{j} = \Phi(\mG_{S^{(1)}_{j}}, \bX_{S_{j}}, \bY_{S_{j}}).$ {\Comment{\footnotesize Sample $J^{(1)}$ datasets and train a GNN on them.}}
\EndFor
\For{$j = 1, \ldots, N^{B}$}
\State $S^{(2)}_{j} \sim  \mS(\mG)$,  $\Phi^{(2)}_{j} = \Phi(\mG_{S^{(2)}_{j}}, \bX_{S^{(2)}_{j}}, \bY_{S^{(2)}_{j}})$. {\Comment{\footnotesize Sample $J^{(2)}$ datasets and train a GNN on them (used to challenge MIA)}} 
%
\EndFor
\For{$k = 1, \ldots, K$}
\State Set $\nu = \nu_{k}$
\State $\hat{\gamma}_{\nu}(1) = \#\{j \in \{1, \ldots, J^{(1)} \}: \nu \in S^{(1)}_{j}\}/J^{(1)}$. {\Comment{\footnotesize Estimate the prior membership probability}}
\State  $\mH_{0} = \{\Phi^{A}_{j}: j \in \{1, \ldots, J^{A}\}, \nu \notin S^{(1)}_{j} \}$,  $\mH_{1} = \{\Phi^{A}_{j}: j \in \{1, \ldots, J^{A}\}, \nu \in S^{(1)}_{j} \}$, {\Comment{\footnotesize Construct the sets for $H_{0}$, $H_{1}$}}
\State Set $A_{k} = B_{k} = 0$, $n_{0, k} = n_{0, k} = 0$
\For{$j \in \{1, \ldots, J^{(2)}\}$}
\State $m = \mathbb{I}(\nu \in S^{2}_{j})$. {\Comment{\footnotesize True membership}}
\State $n_{m, k} \leftarrow n_{m, k} + 1$ {\Comment{\footnotesize True count}}
\State $\hat{m} \gets \textsc{Decision}(\mathcal{I}_\nu, y_\nu,\mH_{0},\mH_{1}, \hat\gamma_{\nu}(1))$ \Comment{\footnotesize Decide on the membership}
\State $A_{k} \leftarrow A_{k} + \mathbb{I}( \hat{m} \neq m = 0)$ {\Comment{\footnotesize type-I error}}
\State $B_{k} \leftarrow B_{k} + \mathbb{I}(\hat{m} \neq m = 1)${\Comment{\footnotesize type-II error}}
\EndFor
\EndFor
\end{algorithmic}
\end{algorithm}

\section{Datasets and models}
\label{sec:datasetsandmodels}
\subsection{Datasets}
\label{sec:datasets}
We conduct experiments on two widely used citation network datasets from the Planetoid benchmark: \cora and \cseer.
Cora consists of $2,708$ scientific publications (nodes) classified into $7$ research topics (classes). The citation network contains $5,429$ edges. Each node is represented by a 1,433-dimensional sparse bag-of-words feature vector indicating the presence of specific terms in the document.

\cseer contains $3,327$ scientific publications (nodes) categorized into $6$ research areas (classes), with $4,732$ citation links (edges). Node features are $3,703$-dimensional sparse bag-of-words vectors derived from document text.

\subsection{\gcn and \gat models}
For the Graph Convolutional Network (GCN)~\citep{kipf2016semi}, the aggregation corresponds to normalized neighborhood feature averaging. 
Let $\tilde{\mA} = \mA + \mathbf{I}$ be the adjacency matrix with added self-loops 
and $\tilde{\mathbf{D}}$ the corresponding degree matrix. 
The layer-wise update can be written as
\begin{align*}
\mathbf{z}_i^{(\ell)}
&= \sum_{j \in \mathcal{N}(i) \cup \{i\}}
\frac{1}{\sqrt{\tilde{d}_i \tilde{d}_j}}
\, \mathbf{x}_j^{(\ell-1)}, \\
\mathbf{x}_i^{(\ell)}
&= \sigma\!\left( \mathbf{W}^{(\ell)} \mathbf{z}_i^{(\ell)} \right),
\end{align*}
where $\tilde{d}_i$ is the degree of node $i$ in $\tilde{\mA}$,
$\mathbf{W}^{(\ell)}$ is a learnable weight matrix,
and $\sigma$ is a non-linear activation function.

For the Graph Attention Network (GAT)~\citep{velivckovic2017graph}, 
aggregation is performed via learned attention weights. 
Specifically,
\begin{align*}
\alpha_{ij}^{(\ell)}
&= \frac{
\exp\!\left(
\operatorname{LeakyReLU}\!\big(
\mathbf{a}^{(\ell)\top}
[\mathbf{W}^{(\ell)} \mathbf{x}_i^{(\ell-1)} \,\Vert\,
 \mathbf{W}^{(\ell)} \mathbf{x}_j^{(\ell-1)}]
\big)
\right)
}{
\sum_{k \in \mathcal{N}(i) \cup \{i\}}
\exp\!\left(
\operatorname{LeakyReLU}\!\big(
\mathbf{a}^{(\ell)\top}
[\mathbf{W}^{(\ell)} \mathbf{x}_i^{(\ell-1)} \,\Vert\,
 \mathbf{W}^{(\ell)} \mathbf{x}_k^{(\ell-1)}]
\big)
\right)
},
\\
\mathbf{z}_i^{(\ell)}
&= \sum_{j \in \mathcal{N}(i) \cup \{i\}}
\alpha_{ij}^{(\ell)} \mathbf{W}^{(\ell)} \mathbf{x}_j^{(\ell-1)}, \\
\mathbf{x}_i^{(\ell)}
&= \sigma\!\left( \mathbf{z}_i^{(\ell)} \right),
\end{align*}
where $\mathbf{a}^{(\ell)}$ is a learnable attention vector and
$\Vert$ denotes concatenation. In practice, multi-head attention is often used, and the outputs of multiple attention heads are concatenated or averaged.

\subsection{Implementation details} \label{sec: Implementation details}
For the \gcn experiments, we use a 2-layer GCN with hidden dimension \(16\), trained for \(100\) epochs. For the GAT experiments, we use a 2-layer \gat trained for 
100 epochs, with a hidden dimension of 4 per attention head in the first layer (with 8 heads in the first layer). 

In the graph sampling pipeline, we set the maximum number of sampled neighbors to \(5\) (for snowball expansion). For the noisy-output setting, we add Gaussian noise to the parameters of the last layer with standard deviation \(0.25\). 

In all the Bayesian estimation tasks, we relied on the hierarchical model in \eqref{eq: joint pdf} and the MHAAR algorithm in Algorithm \ref{alg: MCMC-DP-Est} designed for that hierarchical model, with necessary modifications to suit the methodology for the given privacy definition. Specifically,
\begin{itemize}
\item For BMP-R, we fixed $\varepsilon_{L} = \infty$ and assign the prior $\varepsilon_{R} \sim \mathcal{N}(0, 10)$;
\item For MP, we took $\varepsilon_{\text{MP}} \sim \mathcal{N}(0, 10)$. 
\item For BMP, we took $\varepsilon_{L},\varepsilon_{R} \overset{\text{iid}}{\sim} \mathcal{N}(0, 10)$. 
\end{itemize}
For all tasks, the attack strength parameter is fixed to $s = 0$, i.e., the $(\alpha_{k}, \beta_{k})$-error probabilities of a MIA attack are assumed to be uniformly distributed over their allowed region $\mR_{\eta_{\nu_{k}}}(\varepsilon_{L}, \varepsilon_{R})$. The error counts were modelled as independent Binomials, i.e., $A_{k} \sim \text{Binom}(n_{0, k}, \alpha_{k})$ and $B_{k} \sim \text{Binom}(n_{1, k}, \beta_{k})$, independently, so the parameter $\theta$ is absent from this model. 

We used independent normal random-walk proposals for $\varepsilon_{L}$ and $\varepsilon_{R}$, whose variances are adapted during MCMC iterations to achieve an acceptance rate of $0.23$. In all experiments, the number of auxiliary variables $M$ for $(\alpha_{k}, \beta_{k})$ for each $k$ is taken $M = 50$. To sample auxiliary variables for $(\alpha_{k}, \beta_{k})$, we used 
\[
q_{k}(\alpha^{(i)}_{k}, \beta^{(i)}_{k}) = \text{Beta}(\alpha^{(i)}_{k} | A_{k} + 1, n_{0, k} - A_{k} +1) \text{Beta}(\beta^{(i)}_{k} | B_{k} + 1, n_{1, k} - B_{k} +1).
\]
To estimate the $\varepsilon_{\text{MP}}$ parameter of $\varepsilon_{\text{MP}}$-MP and the $\varepsilon_{R}$ parameter of $\varepsilon_{R}$-BMP-R, we ran Algorithm \ref{alg: MCMC-DP-Est} for $T = 10000$ iterations and the last 5000 samples are considered after a burn-in period of 5000 iterations. 

To jointly estimate the $\varepsilon_{L}, \varepsilon_{R}$ parameters of $(\varepsilon_{L}, \varepsilon_{R})$-BMP, we ran Algorithm \ref{alg: MCMC-DP-Est} for $T = 50000$ iterations, and the last 25000 samples are considered after a burn-in period of 25000 iterations.

\section{Detailed Experimental Results} \label{sec: Detailed Experimental Results}

\subsection{Additional results for estimating MP and BMP-R parameters}
\begin{table*}[ht]
\caption{Hierarchical results table (p5 / p50 / p95 per cell) when Cora is employed with $50\%$ nodes sampled to build train graphs.}
\label{tbl:Cora_50}
\centering
\small
\setlength{\tabcolsep}{4pt}
\renewcommand{\arraystretch}{1.2}

\begin{tabular}{lll|cccc|cccc}
\hline
\multirow{3}{*}{Sampling} & \multirow{3}{*}{Decision} & \multirow{3}{*}{Output}
 & \multicolumn{4}{c|}{\gcn} & \multicolumn{4}{c}{\gat} \\

 &  &  & \multicolumn{2}{c}{full edge} & \multicolumn{2}{c|}{no edges} & \multicolumn{2}{c}{full edge} & \multicolumn{2}{c}{no edges} \\

 &  &  & BMP-R & MP & BMP-R & MP & BMP-R & MP & BMP-R & MP \\
\hline
\multirow{4}{*}{random} & \multirow{2}{*}{weak} & clean  & 4.1/4.4/4.8 & 4.6/4.6/4.8 & 6.7/8.1/10 & 6.8/8/10 & 2.4/2.5/2.5 & 3.6/3.6/3.6 & 6.3/7.7/10 & 6.3/7.5/9.8 \\
 &  & noisy  & 2.8/2.8/2.8 & 3.1/3.1/3.2 & 4.4/4.7/5.3 & 4.4/4.6/4.9 & 1.5/1.5/1.6 & 3.1/3.2/3.3 & 4.6/5/5.6 & 4.4/4.7/5 \\
 & \multirow{2}{*}{strong} & clean  & 3.8/3.9/4.3 & 4.1/4.1/4.2 & 6.7/8/10 & 6.8/8.1/10 & 2.4/2.4/2.5 & 3.6/3.6/3.6 & 6.2/7.7/9.8 & 6.4/7.5/10 \\
 &  & noisy  & 2.8/2.8/2.8 & 3.2/3.2/3.3 & 4.5/4.8/5.5 & 4.4/4.6/4.9 & 2.9/2.9/2.9 & 3.2/3.2/3.2 & 4.6/5.1/5.7 & 4.6/4.8/5.1 \\
\hline
\multirow{4}{*}{snowball} & \multirow{2}{*}{weak} & clean  & 7.2/8.6/11 & 10/11/13 & 7/8.3/11 & 10/11/13 & 7.2/8.5/11 & 10/11/13 & 5.9/6.6/8.2 & 9.7/11/13 \\
 &  & noisy  & 5.7/6.4/8.4 & 9.5/11/13 & 6.3/7.5/9.9 & 9.4/11/13 & 6.1/7/9.2 & 9.8/11/13 & 5.7/6.2/7.6 & 9.6/11/13 \\
 & \multirow{2}{*}{strong} & clean  & 6.8/8/10 & 10/12/14 & 6.7/7.9/10 & 10/11/13 & 6.7/7.9/11 & 10/12/14 & 5.7/6.5/8.3 & 9.9/11/13 \\
 &  & noisy  & 5.6/6.2/7.8 & 9.8/11/13 & 5.5/6/7.2 & 9.8/11/13 & 5.8/6.6/8.4 & 10/11/13 & 5.5/6/7.3 & 9.8/11/13 \\
\hline
\end{tabular}
\end{table*}

\begin{table*}[ht]
\caption{Hierarchical results table (p5 / p50 / p95 per cell) when CiteSeer is employed with $25\%$ nodes sampled to build train graphs.}
\label{tbl:CiteSeer_25}
\centering
\small
\setlength{\tabcolsep}{4pt}
\renewcommand{\arraystretch}{1.2}

\begin{tabular}{lll|cccc|cccc}
\hline
\multirow{3}{*}{Sampling} & \multirow{3}{*}{Decision} & \multirow{3}{*}{Output}
 & \multicolumn{4}{c|}{\gcn} & \multicolumn{4}{c}{\gat} \\

 &  &  & \multicolumn{2}{c}{full edge} & \multicolumn{2}{c|}{no edges} & \multicolumn{2}{c}{full edge} & \multicolumn{2}{c}{no edges} \\

 &  &  & BMP-R & MP & BMP-R & MP & BMP-R & MP & BMP-R & MP \\
\hline
\multirow{4}{*}{random} & \multirow{2}{*}{weak} & clean  & 6.1/7.5/10 & 7.6/8.9/11 & 7.2/8.6/11 & 8.6/9.9/12 & 5.9/7.3/9.7 & 7/8.3/11 & 7.5/8.8/11 & 9/10/12 \\
 &  & noisy  & 3.1/3.2/3.5 & 4.2/4.2/4.3 & 4.8/5.6/7.7 & 5.7/6.1/7.1 & 3/3.1/3.3 & 4.5/4.5/4.6 & 4.3/4.8/6 & 5.4/5.8/6.6 \\
 & \multirow{2}{*}{strong} & clean  & 6.9/8.2/11 & 7.8/9.2/11 & 8/9.4/12 & 9.1/10/13 & 6.6/8.1/11 & 7.6/8.9/11 & 8.2/9.6/12 & 9.3/11/13 \\
 &  & noisy  & 4.3/5.4/7.7 & 4.8/5/5.4 & 5.9/7.3/9.7 & 6.5/7.3/9.4 & 3.9/4.5/6.4 & 4.5/4.6/4.9 & 5.5/6.6/9 & 6/6.5/7.5 \\
\hline
\multirow{4}{*}{snowball} & \multirow{2}{*}{weak} & clean  & 5.1/6.4/9.5 & 11/13/15 & 5.5/6.9/9.4 & 11/13/15 & 3.8/4.3/6.1 & 11/12/14 & 3.5/3.8/4.5 & 11/12/14 \\
 &  & noisy  & 3.5/3.9/4.6 & 11/12/14 & 3.7/4/4.8 & 11/12/14 & 2.9/3.2/3.5 & 11/12/14 & 2.9/3/3.2 & 11/12/14 \\
 & \multirow{2}{*}{strong} & clean  & 6.2/7.7/10 & 11/13/14 & 6.3/7.7/10 & 12/13/15 & 5.4/6.8/9.6 & 11/13/15 & 4.9/6.3/9.3 & 11/13/15 \\
 &  & noisy  & 4.8/6.4/9 & 11/12/15 & 5/6.4/9 & 11/13/14 & 4.7/6.2/8.8 & 11/12/15 & 3.8/5.2/8.3 & 11/12/14 \\
\hline
\end{tabular}
\end{table*}

\begin{table*}[ht]
\caption{Hierarchical results table (p5 / p50 / p95 per cell) when CiteSeer is employed with $50\%$ nodes sampled to build train graphs.}
\label{tbl:CiteSeer_50}
\centering
\small
\setlength{\tabcolsep}{4pt}
\renewcommand{\arraystretch}{1.2}

\begin{tabular}{lll|cccc|cccc}
\hline
\multirow{3}{*}{Sampling} & \multirow{3}{*}{Decision} & \multirow{3}{*}{Output}
 & \multicolumn{4}{c|}{\gcn} & \multicolumn{4}{c}{\gat} \\

 &  &  & \multicolumn{2}{c}{full edge} & \multicolumn{2}{c|}{no edges} & \multicolumn{2}{c}{full edge} & \multicolumn{2}{c}{no edges} \\

 &  &  & BMP-R & MP & BMP-R & MP & BMP-R & MP & BMP-R & MP \\
\hline
\multirow{4}{*}{random} & \multirow{2}{*}{weak} & clean  & 7.2/8.6/11 & 8.3/9.7/12 & 8.5/9.9/12 & 8.8/10/12 & 6.5/8/10 & 7.1/8.4/11 & 7.7/9/11 & 8/9.5/12 \\
 &  & noisy  & 4.3/4.6/5 & 4.5/4.8/5.1 & 5.7/6.3/7.7 & 5.7/6.1/6.9 & 3.7/3.8/4 & 4.2/4.2/4.3 & 4.6/4.9/5.3 & 4.6/4.8/5 \\
 & \multirow{2}{*}{strong} & clean  & 7.2/8.5/11 & 8.3/9.7/12 & 8.5/9.8/12 & 8.8/10/12 & 6.6/8/11 & 7.1/8.4/11 & 7.6/9/11 & 7.9/9.2/11 \\
 &  & noisy  & 4.4/4.6/5.1 & 4.6/4.8/5.1 & 5.7/6.3/7.7 & 5.7/6.1/6.9 & 3.8/3.8/3.9 & 4.4/4.4/4.4 & 4.6/4.9/5.3 & 4.5/4.8/5 \\
\hline
\multirow{4}{*}{snowball} & \multirow{2}{*}{weak} & clean  & 8.7/10/12 & 12/13/15 & 7.3/8.7/11 & 12/13/15 & 7.1/8.4/11 & 12/13/15 & 5.5/6.2/8.1 & 11/12/14 \\
 &  & noisy  & 5.5/6.3/8.4 & 11/12/14 & 5.1/5.4/6.6 & 11/12/14 & 4.6/4.9/5.3 & 11/12/14 & 4.3/4.5/4.9 & 11/12/14 \\
 & \multirow{2}{*}{strong} & clean  & 9/10/13 & 12/13/15 & 7.5/8.9/11 & 12/13/15 & 7.3/8.5/11 & 12/13/15 & 5.7/6.5/8.9 & 11/13/15 \\
 &  & noisy  & 5.7/6.6/8.7 & 11/12/14 & 5/5.6/7.4 & 11/12/14 & 4.7/5.1/6.3 & 11/13/14 & 4.3/4.5/4.9 & 11/12/14 \\
\hline
\end{tabular}
\end{table*}

\end{document}